\documentclass[twoside]{article}
\usepackage[preprint]{aistats2026}

\usepackage{amsmath,amsfonts,bm}

\def\eqref#1{equation~\ref{#1}}

\def\1{\bm{1}}

\def\vgamma{{\bm{\gamma}}}

\def\vell{{\bm{\ell}}}

\def\vx{{\bm{x}}}

\def\vz{{\bm{z}}}

\DeclareMathAlphabet{\mathsfit}{\encodingdefault}{\sfdefault}{m}{sl}
\SetMathAlphabet{\mathsfit}{bold}{\encodingdefault}{\sfdefault}{bx}{n}

\def\sZ{{\mathbb{Z}}}

\usepackage{graphicx}
\usepackage{subcaption}

\usepackage{hyperref}
\usepackage{url}
\usepackage{todonotes}
\usepackage{array}
\usepackage{ifthen}
\usepackage{amsthm}

\usepackage{import}

\usepackage[round]{natbib}
\renewcommand{\bibname}{References}
\renewcommand{\bibsection}{\subsubsection*{\bibname}}

\newcommand{\mse}{MSE}
\newcommand{\snr}{SNR}

\graphicspath{{figures/}}

\newtheorem{prop}{Proposition}

\makeatletter
\newcommand{\ostar}{\mathbin{\mathpalette\make@circled\star}}
\newcommand{\make@circled}[2]{%
  \ooalign{$\m@th#1\smallbigcirc{#1}$\cr\hidewidth$\m@th#1#2$\hidewidth\cr}%
}
\newcommand{\smallbigcirc}[1]{%
  \vcenter{\hbox{\scalebox{0.77778}{$\m@th#1\bigcirc$}}}%
}
\makeatother

\newcommand{\transp}{\text{T}}
\newcommand{\bigO}{\mathcal{O}}

\newcommand{\Angstr}{\text{Å}}

\begin{document}

\twocolumn[

\aistatstitle{Joint Denoising of Cryo-EM Projection Images using Polar Transformers}

\aistatsauthor{Joakim Andén \And Justus Sagemüller}

\aistatsaddress{KTH Royal Institute of Technology, \\ Flatiron Institute, Simons Foundation \And KTH Royal Institute of Technology}

]

\begin{abstract}
Many imaging modalities involve reconstruction of unknown objects from collections of noisy projections related by random rotations.
In one of these modalities, cryogenic electron microscopy~(cryo-EM), the extremely low signal-to-noise ratio~(SNR) makes integration of information from multiple images crucial.
Existing approaches to cryo-EM processing, however, either rely on handcrafted priors or apply deep learning only on select portions of the pipeline, such as particle picking, micrograph denoising, or refinement.
A fully end-to-end reconstruction approach requires a neural network architecture that integrates information from multiple images while respecting the rotational symmetry of the measurement process.
In this work, we introduce the \emph{polar transformer}, a new neural network architecture that combines polar representations and transformers along with a \emph{convolutional attention} mechanism that preserves the rotational symmetry of the problem.
We apply it to the particle-level denoising problem, where it is able to learn discriminative features in the images, enabling optimal clustering, alignment, and denoising.
On simulated datasets, this achieves up to a $2\times$ reduction in mean squared error~(MSE) at a signal-to-noise ratio~(SNR) of $0.02$, suggesting new opportunities for data-driven approaches to reconstruction in cryo-EM and related tomographic modalities.
\end{abstract}

\section{INTRODUCTION}
\label{sec:intro}
Many scientific imaging problems involve reconstructing a single unknown object from multiple noisy projections.
This problem arises in tomography and related modalities, such as X-ray tomography, diffraction imaging, electron tomography, synthetic aperture radar, and single-particle cryogenic electron microscopy~(cryo-EM).
In all these settings, the very low signal-to-noise ratio~(SNR) means that achieving an accurate reconstruction requires combining information from multiple viewing directions while accounting for the rotational symmetry of the measurement process.
Among these modalities, cryo-EM provides a particularly challenging and high-impact instance of this problem.

Cryo-EM is a molecular imaging technique that uses transmission electron microscopes to analyze the structure of molecules, typically proteins of interest for biological or medical research~\citep{Dubochet1982CryoEM}.
This involves freezing a sample containing millions of copies of a single molecule in a thin layer of vitreous ice, which is then imaged in an electron microscope.
The 3D structure of the molecule is then reconstructed from the resulting 2D projection images, recorded as \emph{micrographs}.
The reconstruction process involves several steps, such as first extracting images of individual particles from the micrographs, denoising these images, constructing a low-resolution initial reference map, and then refining this into the final 3D density map~\citep{Frank1996CryoEM}.

An important challenge of cryo-EM imaging is the high noise level.
Because electron dose must be kept low to avoid radiation damage, each projection contains only a weak signal and a large amount of shot noise~\citep{Bendory2020SPCEM}.
Denoising is therefore a critical step in the cryo-EM processing pipeline.
For example, it is needed for visual inspection~\citep{Singer2020Methods}, detecting contamination (outliers)~\citep{Bhamre2016Denois}, generating templates for particle picking~\citep{Singer2020Methods}, identifying high-quality projection images~\citep{Scheres2015ClassAvg}, and as preprocessing for certain ab initio reconstruction methods~\citep{Singer2011CommonLines}.
More broadly, the central methodological challenge goes beyond denoising: integrating information from multiple viewing directions while respecting rotational symmetry is central to cryo-EM processing.
A neural network architecture with these properties is therefore valuable for developing data-driven algorithms for cryo-EM reconstruction.

Standard approaches to cryo-EM denoising, such as Wiener filters~\citep{Frank1996SPIDER} and class averaging~\cite{Heel1984MVClassif}, do not learn from data and therefore rely on handcrafted priors of image structure.
As such, they are limited in accuracy, especially at high noise levels.
While deep learning methods have been applied with success, these have been limited to denoising entire micrographs~\citep{Bepler2019TopazDenoise, Buchholz2019CryoCARE, Palovcak2020Denoising} and do not explicitly incorporate rotational symmetry, limiting their applicability for particle-level denoising and other symmetry-aware tasks central to tomographic reconstruction.

This work introduces the \emph{polar transformer}, a new neural network architecture that rotationally aligns and integrates information from multiple images using a \emph{convolutional attention} mechanism that guarantees overall rotational equivariance.
We thus extend the traditional class averaging method into a data-driven, symmetry-preserving framework that learns discriminative features from noisy images to improve clustering and alignment.
While we demonstrate this architecture in the context of particle-level denoising as a proof of concept, its design is general and applicable to other problems in cryo-EM or tomographic reconstruction more broadly.
Our numerical evaluation focuses on simulated data where an exact ground truth is available, allowing for precise quantitative evaluation.
Extending the polar transformer to experimental data is a natural next step.
On simulated datasets, the polar transformer successfully clusters, aligns, and denoises sets of cryo-EM particle images, achieving up to a $2\times$ reduction in relative mean squared error~(MSE) compared to state-of-the-art denoisers at a signal-to-noise ratio~(SNR) of $0.02$.
Reconstructions based on the denoised images similarly show improved accuracy compared to other denoising methods, suggesting a path forward for end-to-end, data-driven approaches to cryo-EM reconstruction.

The rest of this paper is structured as follows.
Section~\ref{sec:related} reviews the existing literature on cryo-EM image denoising and symmetry-aware neural architectures.
Then Section~\ref{sec:polar} describes the construction of the polar transformer and some of its properties.
Finally, Section \ref{sec:results} presents our numerical evaluation.

\section{RELATED WORK}
\label{sec:related}

In general, DNNs have been used in the cryo-EM reconstruction pipeline to provide stronger priors during the 3D refinement stage.
For example, \citet{Kimanius2021Prior, Kimanius2024Prior} replaced the weak Gaussian prior used in the RELION~\citep{Scheres2012Relion} software with a learned prior.
This resulted in an increase in accuracy, but also showed signs of hallucination.
More importantly, these methods operate only at the last stage of cryo-EM reconstruction, where gains are inherently limited by prior stages of the pipeline.

Another major line of work in cryo-EM uses DNNs as general-purpose function approximators, optimized in an unsupervised manner to each cryo-EM dataset.
These are typically used to parametrize the 3D reconstruction~\citep{Zhong2021CryoDRGN, Gupta2021CryoGAN, Schwab2024DynaMight, Nashed2021CryoPoseNet, Levy2022CryoFIRE} or the set of 2D projections~\citep{Bibas2021RotInvFeat, Kwon2023RotInvRepr, Nasiri2022VAE}.
Consequently, the neural networks typically do not impose priors beyond the inductive bias of the particular architecture.
For the 3D reconstruction case, these methods also do not incorporate the rotational symmetry of the cryo-EM problem.
This is in contrast with the context of the present work, where DNNs are trained to form general-purpose estimators by encoding prior information extracted from large collections of cryo-EM data.

Due to the high noise level, the denoising problem in cryo-EM has a long history.
Classical approaches include low-pass filters (referred to as ``binning'')~\citep{Bartesaghi2018Atomic}, stationary Wiener filters~\citep{Frank1996SPIDER, Tang2007Eman2, Sindelar2011Wiener} or generalized Wiener filters calculated in a steerable basis~\cite{Bhamre2016Denois}.
While effective to some extent and theoretically well understood, these methods are linear and only optimal under Gaussian priors~\citep{Kay1993Fundamentals}, leading to overly blurred images at high noise levels.

Another approach that exploits the geometry of the problem is class averaging~\citep{Heel1984MVClassif, Park2014ClassAvg, Zhao2014ClassAvg, Scheres2015ClassAvg}.
Here, images are clustered by viewing direction, aligned, and averaged.
This reduces the noise level with less blurring compared to filter-based methods, but requires hundreds of images per cluster.
More importantly, clustering and alignment are themselves sensitive to noise and thus limited in what SNR they can address.

Deep learning methods for denoising are able to encode more sophisticated priors but do not exploit the geometry of the problem as effectively.
Standard approaches here are DnCNNs~\citep{Zhang2017DnCNN} and U-Nets~\citep{Ronneberger2015Unet, GurollaRamos2021Unet}.
While U-Nets have been applied to cryo-EM denoising, this has been restricted to whole micrographs~\citep{Bepler2019TopazDenoise, Palovcak2020Denoising}.
While this avoids the need for particle picking, it is limited by considering only one micrograph at a time and does not incorporate the rotational symmetry of the imaging model.

Our proposed approach can be seen as a data-driven extension of class averaging.
We replace handcrafted similarity measures for clustering and alignment with a neural architecture that learns discriminative features of the images while preserving the rotational symmetry of the problem.
This bridges shallow, geometry-aware methods and deep neural priors, enabling symmetry-preserving, data-driven integration of cryo-EM images.

\section{POLAR REPRESENTATIONS, CNNS, AND TRANSFORMERS}
\label{sec:polar}

To construct a neural network architecture capable of aggregating information from multiple rotated images, we first need a representation that allows for rotating images in an efficient and stable manner.
This will also ensure rotational equivariance, allowing us to preserve the symmetries of the problem by construction.

\subsection{Polar Representation}

\begin{figure*}[t]
    \begin{center}
        \begin{subfigure}{25mm}
            \centering
            \includegraphics[height=24mm]{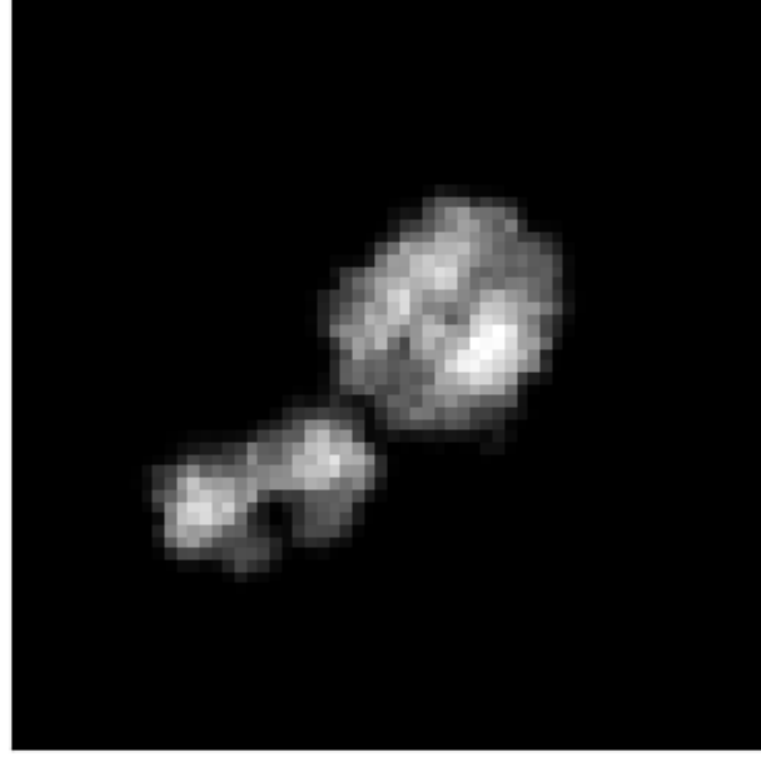}
            \subcaption{\label{fig:polar-decomp-orig}}
        \end{subfigure}%
        \begin{subfigure}{25mm}
            \centering
            \includegraphics[height=24mm]{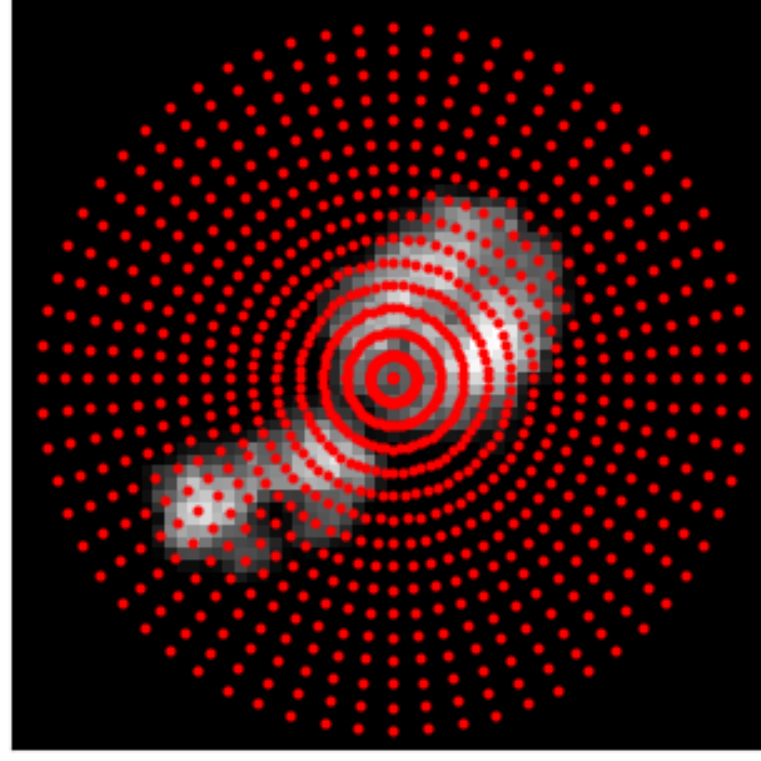}
            \subcaption{\label{fig:polar-decomp-grid}}
        \end{subfigure}%
        \begin{subfigure}{95mm}
            \centering
            \includegraphics[height=24mm]{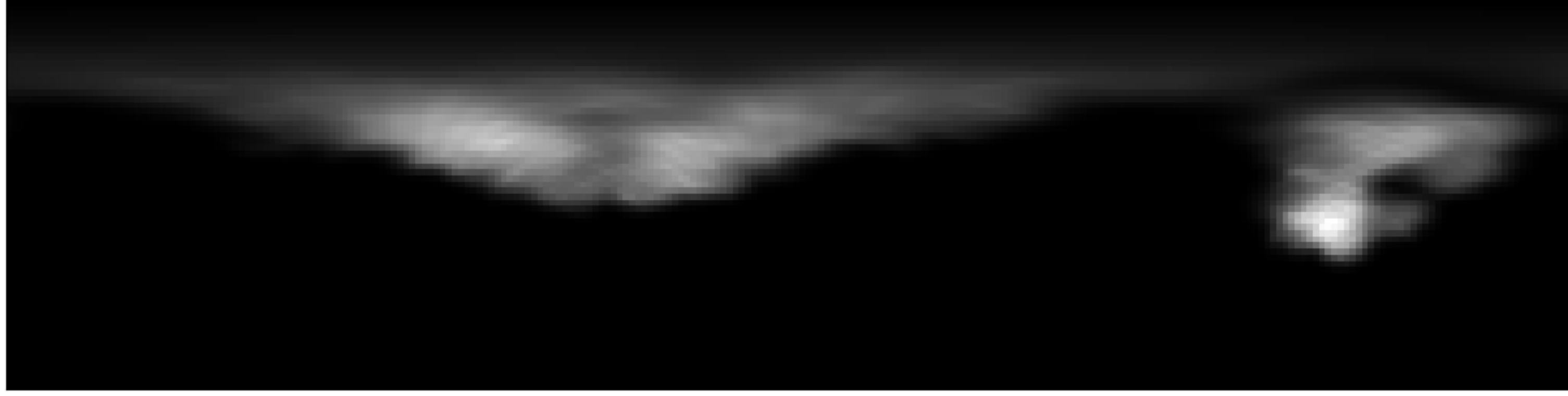}
            \subcaption{\label{fig:polar-decomp-polar}}
        \end{subfigure}%
        \begin{subfigure}{25mm}
            \centering
            \includegraphics[height=24mm]{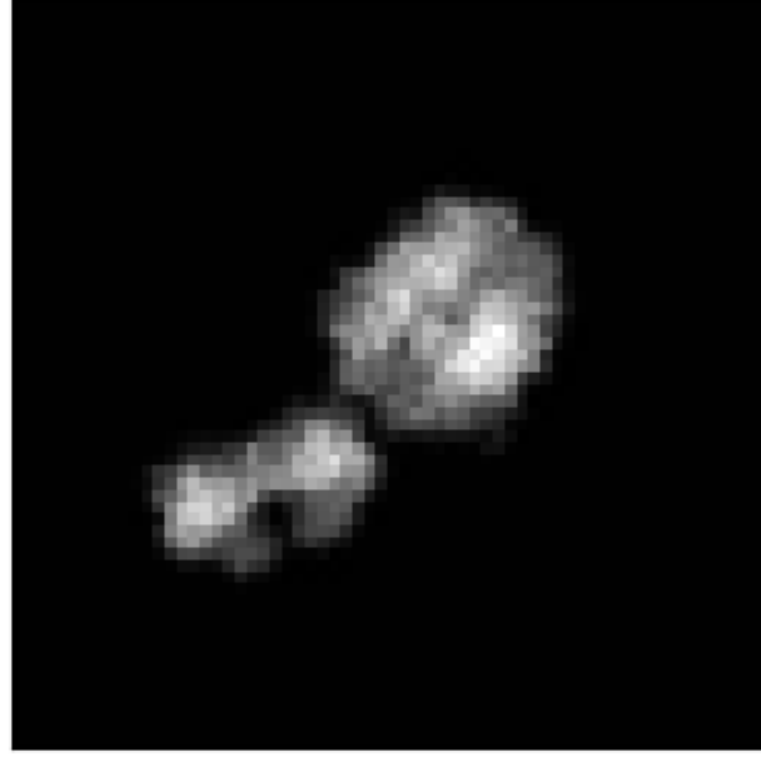}
            \subcaption{\label{fig:polar-decomp-recon}}
        \end{subfigure}
    \end{center}
    \caption{\label{fig:polar-decomp} (a) A simulated $64\times 64$ projection image of PDB ID 2pkq. (b) The projection image with a polar grid superimposed (downsampled by $4\times$ for visualization). (c) The polar representation (horizontal axis: angles; vertical axis: radii). (d) The reconstructed image with relative MSE of $3.5 \cdot 10^{-4}$.}
\end{figure*}

To handle rotations efficiently, we represent the images on a polar grid, where rotations correspond to simple angular shifts.
The original images are given in the Cartesian domain, with pixels on an $L\times L$ grid.
Mapping those to a polar grid by naive interpolation is prone to regularity problems at the pole, which can however largely be avoided by the following techniques.

First, let us assume that the pixel size of the images corresponds to $2/L$, so that the $L\times L$ Cartesian grid spans the square $[-1, 1]^2$~(see Figure~\ref{fig:polar-decomp-orig}).
We now place an $N \times M$ polar grid $(u_{nm}, v_{nm})$ inside this square.
The $N$ radial nodes are given by a Gauss--Jacobi quadrature rule with quadrature weights $w_n$, while the $M$ angular nodes are uniformly spaced on the circle (for more information, see~Appendix~\ref{sec:convoApprox}).
Typically, we choose $N$ and $M$ to be $L$ and $4L$, respectively.

To map our Cartesian images with pixel values $x[i,j]$ to the polar domain, we place a Gaussian radial basis function~(RBF) at each of the grid points and compute inner products.
Let us assume that $L$ is even and that the Cartesian image $x$ is indexed by $-L/2, \ldots, L/2 - 1$ along both axes.
The polar decomposition $Px[n, m]$ is then given by~(see Figure~\ref{fig:polar-decomp-polar})
\begin{equation}
    \label{eq:polar-decomp}
    \frac{\sqrt{w_n}}{Z} \sum_{i,j=-L/2}^{L/2-1} x[i, j] e^{-\frac{(u_{nm} - 2i/L)^2 + (v_{nm} - 2j/L)^2}{2b^2}}
\end{equation}
for $n = 0, 1, \ldots, N-1$ and $m = 0, 1, \ldots, M-1$, where $b$ is a bandwidth factor typically set to $1/L$ and the normalization factor $Z$ is given by $Z = \sqrt{2\pi^3 M L^2 b^4}$.
Each polar coefficient is therefore given by a Gaussian-weighted interpolation of nearby Cartesian pixel values.

This representation enjoys a number of useful properties.
For one, the original image can be provably recovered to arbitrary accuracy by an appropriate choice of parameters $N$, $M$, and $b$ (see Figure~\ref{fig:polar-decomp-recon} and Appendix~\ref{sec:convoApprox}).
The representation is also an approximate isometry when restricted to smooth images (see Appendix~\ref{sec:sumPhi}).
Because of this, the polar representation will not significantly distort the geometry of the original data manifold.

The decomposition also enjoys stability in the sense that the operator norm $\|P\|$ can be shown to be approximately one (see Appendix~\ref{sec:sumPhi}).
In other words, small changes in $x$ result in small changes in $Px$.
This is in contrast with other polar representations, such as those obtained by nearest-neighbor or linear interpolation, which can introduce artifacts for small changes in the Cartesian image~(see numerical experiments in Appendix~\ref{sec:sumPhi}).
This lack of artifacts simplifies the training of neural networks in the representation, improves their equivariance properties, and increases their robustness to noise.
Finally, both $P$ and its inverse can be applied efficiently, with computational complexity $\bigO(L^2)$ and $\bigO(L^2 \log L)$, respectively (see Appendix~\ref{sec:compComplexity}).

\subsection{Polar CNNs}
\label{sec:polcnns}

A crucial property of the polar representation is that it commutes with rotation.
Indeed, we have the relationship $P R_{\gamma_\ell} \approx S_\ell P$, where $S_\ell$ denotes circular shift along the second axis by $\ell$ and we recall that $\gamma_\ell = 2\pi \ell / M$.
Rotation in the Cartesian domain thus corresponds to translation along the angular axis in the polar domain.
Due to the stability of the polar decomposition $P$, off-grid rotations will give sensible polar representations in the form of subgrid shifts.

Equivariance to rotations in the Cartesian domain can therefore be achieved by equivariance to angular translations in the polar domain.
The space of linear operators equivariant to such translations is spanned by \emph{angular convolutions}: operators of the form
\begin{equation}
    z \ostar h[n, m] = \sum_{p=0}^{N-1} \sum_{q=-Q}^{Q} z[p, m - q] h[n, p, q],
\end{equation}
where $z$ is a polar representation of size $N\times M$ and $Q$ is the (angular) width of $h$ (typically in the range $1$--$3$).
In other words, we convolve only along the angular axis, treating the radial axis as a fully connected layer.
We refer to $h$ as an \emph{angular filter}, with two radial indices $n$, $p$ and one angular index $q$.

As defined above, the filter acts globally along the radial axis.
We can enforce localization of the filter along the radial axis by restricting its support.
Specifically, we set $h[n, p, q] = 0$ whenever $|n - p| > W$, for some radial width $W$ (typically in the range $1$--$3$).

The angular filter can be trivially generalized to multiple channels and used as a linear layer in a DNN.
Since the polar representation is essentially a reparametrization of the original Cartesian image, standard layers from CNNs can be applied to the output of the angular convolutions, such as ReLUs, avoiding less natural spectral nonlinearities required for other equivariant networks~\citep{Kondor2018CGNets}.
The above construction also avoids costly transformations between an equivariant basis and a natural image basis~\citep{Cohen2018Spherical}.
Finally, we note that this can also be formulated as an $\text{SO}(2)$-CNN in the formalism of~\citet{Kondor2018Equivariance}.

We shall refer to a DNN consisting of the angular convolutions described above as a \emph{polar CNN}.

\subsection{Polar Transformer}
\label{sec:transformer}

We are now ready to define the polar transformer architecture.
Transformers are a natural architecture for aggregating disparate sources of information, originally introduced in the field of language models~\citep{Vaswani2017Attention, Devlin2018Bert, Brown2020GPT3}.
Here, a sequence of \emph{tokens} is used to generate key, query, and value vectors that are combined using what is known as an \emph{attention mechanism}~\citep{Vaswani2017Attention}.
In language models, tokens are parts of words, but they can be any information carrier, such as an image.

Let us consider a multi-channel polar image $z[c, n, m]$, where $c = 0, 1, \ldots, C-1$ is the channel index, $n = 0, 1, \ldots, N-1$ is the radial index, and $m = 0, 1, \ldots, M-1$ is the angular index.
These are used as input to three polar CNNs, denoted $f^{\text{key}}_\theta$, $f^{\text{query}}_\theta$, and $f^{\text{value}}_\theta$, corresponding to the key, query, and value networks, respectively.
Here, $\theta$ denotes the weights of the mechanism (i.e., the weights of the key, query, and value networks).

While the standard attention mechanism may be useful for combining information in the images by clustering them appropriately, it has two important drawbacks: it cannot align the images rotationally and it is not equivariant to rotations.
In fact, it turns out that these problems are closely related -- an attention mechanism that is rotationally equivariant must also perform alignment.

Now suppose we have a set $\vz = (z_0, z_1, \ldots, z_{K-1})$ of $K$ multi-channel polar images, which we process using the above networks to obtain key, query, and value vectors.
To see how this can be achieved, let us consider an augmented form of the standard attention mechanism, where each key is rotated by an arbitrary angle $\gamma_\ell$, that is, shifted by $\ell$ along the angular axis.
We then have a set of rotated attention coefficients $\alpha_{k,k'}^{(\ell)}(\vz)$ given by
\begin{equation}
    \sigma\left(\sum_{c, n, m} f^{\text{query}}_\theta(z_k)[c, n, m] S_\ell^{} f^{\text{key}}_\theta(z_{k'})[c, n, m] \right),
\end{equation}
for $\ell = 0, 1, \ldots, M-1$, where the softmax $\sigma$ is taken over indices $k'$ and $\ell$ (the softmax also includes a scaling factor of $1/\sqrt{CNM}$).
We then apply a corresponding rotation to the value vectors when performing the linear combination, to obtain
\begin{equation}
    (f_\theta^{\text{ang-attention}}(\vz))_k = \sum_{\ell=0}^{M-1} \sum_{k'=0}^{K-1} \alpha_{k,k'}^{(\ell)}(\vz) S_\ell^{} f^{\text{value}}_\theta(z_{k'})
\end{equation}
for $k = 0, 1, \ldots, K-1$.
While this would seem to incur a significant computational expense, both calculations can be written as convolutions along the angular axis.
This allows efficient implementation using fast Fourier transforms~(FFTs), resulting in a computational complexity of $\bigO(K^2 L^2 \log L)$.

Since it performs attention along the angular axis, we refer to this as the \emph{angular attention mechanism}.
It gives the network the ability to align images rotationally and combine them accordingly.
It also satisfies an extended equivariance property.
Let $S_{\vell}$ denote the joint angular shifting operator for $K$ images by $\vell \in \sZ^K$ such that $(S_{\vell} z)_k = S_{\ell_k} z_k$.
It can then be shown that
\begin{equation}
    f^{\text{ang-attention}}_\theta(S_\vell z) = S_\vell f^{\text{ang-attention}}_\theta(z).
\end{equation}
In the Cartesian domain, this means that we have joint equivariance to rotation of the individual images in the set, i.e., the network is $\text{SO}(2)^K$-equivariant.

By stacking one or multiple angular attention blocks together with polar CNNs for pre- and post-processing, we obtain a \emph{polar transformer}.
Note that here, each token is a whole image, as opposed to models such as the Vision Transformer~\citep{Dosovitskiy2020VIT}, where each token corresponds to an image patch.
There is also no positional encoding, ensuring that the model is invariant to permutations of the input image set.

\section{RESULTS}
\label{sec:results}

To evaluate the architectures proposed above, we conduct numerical experiments on denoising simulated data.
These show how the polar CNN and polar transformer are able to achieve very low denoising errors despite high noise levels, significantly outperforming baseline methods.

\subsection{Problem Setup}
\label{sec:setup}
The denoising problem in cryo-EM aims to recover a clean image $x \in X$ from a noisy version $y \in X$.
The clean images are assumed to be tomographic projections of a 3D potential density from a biological macromolecule.
These are calculated using line integrals of that density along some chosen viewing direction.
Following this, the images are filtered using a contrast transfer function~(CTF) to simulate the optical properties of the microscope, and shifted randomly.
We assume that these are individual particle projections (i.e., we are past the particle-picking stage).
To simulate the noisy projections $y$ from the clean image $x$, we then use additive Gaussian noise, a common assumption in many reconstruction methods.

While simple, the above model has proven to be relatively accurate~\citep{Frank1996CryoEM, Vulovic2013ImageFormation} and is used in a range of cryo-EM reconstruction methods~\citep{Barnett2017Marching, Scheres2012Relion, Punjani2017CryoSPARC, Zhong2021CryoDRGN}.
For comparison with data incorporating more realistic microscopy effects, see \autoref{sec:moreResults}.

We will now consider three distinct denoising tasks.

\paragraph{Single images}
The simplest setting is to consider each projection separately, potentially combining it with information from other projections later on.
We thus have a mapping $f: X \rightarrow X$ such that $x \approx f(y)$.
We call this the \emph{single-image denoising task}.
This is the setting of Wiener filters, which are trained on a large set of images and then applied separately to each noisy image.

One desirable property of $f$ is equivariance to 2D rotations.
This follows from the fact that the probability distribution of the images (both clean and noisy) is invariant to in-plane rotations (one image is as likely to appear in the data as a rotated copy of that image).
Letting $R_\gamma: X \rightarrow X$ denote the in-plane rotation of an image by $\gamma \in [0, 2\pi)$, we thus require
\begin{equation}\label{eq:indivProjEquivar}
  f(R_\gamma x) = R_\gamma(f(x)).
\end{equation}

\paragraph{Directional sets}

More potent processing is possible when using a \emph{set} of \(K\) images, \(\vx\in X^K\), representing a selection of projections picked from a microscopy sample.
In particular, we shall assume that these images all come from the same viewing direction and that they only differ by in-plane rotation (in other words, for each $i, j$, there is $\gamma_{ij}$ such that $x_i = R_{\gamma_{ij}} x_j$).
We call this the \emph{directional set denoising task}.
This situation arises, for example, during class averaging when the clustered images are to be aligned and averaged.

For this task, we extend the equivariance requirement so that each image in the set is rotated by a different angle.
In other words, we have the mapping $R_\vgamma: X^K \rightarrow X^K$ for $\vgamma \in [0, 2\pi]^K$ such that $(R_\gamma \vx)_i = R_{\gamma_i} x_i$.
The equivariance relation in this case is
\begin{equation}
    \label{eq:dirSetEquivar}
    f(R_\vgamma \vx) = R_\vgamma(f(\vx)).
\end{equation}

\paragraph{General sets}

In the general case, we cannot assume the images have already been clustered by viewing direction.
Instead, we rely on the fact that any sufficiently large set already contains such clusters of images that are the same up to rotation.
This is the underlying idea of class averaging and represents the most typical denoising setting.
We will refer to this as the \emph{general set denoising task}.
The challenge here is to identify the similar images in the set (up to rotation) and denoise them jointly.
Again, the denoiser takes the form $f: X^K \rightarrow X^K$ and satisfies the equivariance property (\eqref{eq:dirSetEquivar}), but we do not require that clean images are the same up to rotation.

\subsection{Neural Networks}

Two neural network architectures are used in the experiments.
The first is a simple polar CNN consisting of a Cartesian-to-polar layer, 25 angular convolutional layers with $C = 8$ channels each, each followed by a group normalization layer (with four groups)~\citep{Wu2018GroupNorm} and a ReLU nonlinearity.
Finally, the polar representation was converted back to the Cartesian domain.
The convolutional layers had an angular kernel width of $Q = 5$ and a radial kernel width of $W = 5$.

The other network is a polar transformer, starting with a Cartesian-to-polar layer followed by a polar CNN preprocessing network of depth $5$ applied individually to each image in the set.
This was followed by an angular attention module where the key and query networks consisted of the same polar CNN of depth $17$ with shared weights (the value network was the identity).
Finally, the attention output was postprocessed in a $9$-layer polar CNN before conversion to Cartesian.
The convolutional layers in the CNNs had the same configuration for the polar CNN, but the pre- and post-processing networks had $C = 8$ while the key- and query-network had $C = 16$.
The results are relatively stable to changes in these parameters.
For example, we observe similar performance when stacking multiple attention blocks or allowing for more flexibility in the key, query, and value networks.

For comparison, we also trained a DnCNN and a U-Net.
The DnCNN model follows \citet{Zhang2017DnCNN}, while the U-Net is the one used in the Topaz pipeline~\citep{Bepler2019TopazDenoise}, with added group normalization (with four groups) to avoid explosion of weights at low SNR.
For schematics showing the architectures of the networks used, see Appendix~\ref{sec:architectureSchematics}.

\subsection{Data, Training, and Testing}
\label{sec:experimentData}
Both architectures were trained on simulated data obtained from \(5\,000\) randomly selected molecular structures downloaded from the PDB~\citep{PDB2018}.
Each molecule is projected through \(1\,000\) different viewing directions yielding a total of \(5\,000\,000\) clean projection images of size \(64\times64\), corresponding to a pixel size of \(4~\Angstr\).
The same process was used for another \(100\) molecules, resulting in a testing set of \(100\,000\) clean images.

In the directional set denoising task, each clean image was then rotated randomly, resulting in $K = 8$ different clean images.
For general set denoising, two clean images from one molecule corresponding to two different viewing directions were picked, each rotated to yield $8$ copies, for a total of $2K = 16$ images.
In both cases, the whole image set was used as input to the transformer model.
Note that the same original set of clean images is used for training and testing in all three denoising tasks.
They only differ in how they are presented to the model (as individual images or as sets of images).

To generate the noisy images for both training and testing, we first apply a CTF of random defocus (between $1.5~\text{\textmu m}$ and $2.5~\text{\textmu m}$) and a random shift (of maximum $8$ pixels) before adding Gaussian noise.
In addition to Gaussian noise, we also evaluated the architecture for shot noise models in the form of Poisson noise~(see Appendix~\ref{sec:poissonNoiseDetails}).

All architectures were trained to minimize the MSE with respect to the clean images using the Adam optimizer~\citep{Kingma2017Adam} at a learning rate of $10^{-3}$ and batch size $128$.
Convergence was typically obtained after ten epochs for all models.
Training time per epoch was around $1.3$ hours (polar CNN and U-Net), $1.8$ hours (DnCNN), and $12.5$ hours (polar transformer) on an A100 GPU with 40 GB of memory.

\newcommand{\attentionBNMeanMSE}[1]{\ifthenelse{\equal{#1}{"0.03125_noisvarb0.75"}}{\(0.057\)}{\ifthenelse{\equal{#1}{0.03125_poissonic0.5}}{\(0.044\)}{\ifthenelse{\equal{#1}{0.03125}}{\(0.055\)}{\ifthenelse{\equal{#1}{0.0625}}{\(5417.224\)}{\ifthenelse{\equal{#1}{0.125}}{\(0.020\)}{\ifthenelse{\equal{#1}{"0.25_noisvarb0.75"}}{\(0.013\)}{\ifthenelse{\equal{#1}{0.25}}{\(0.013\)}{}}}}}}}}
\newcommand{\attentionBNStdMSE}[1]{\ifthenelse{\equal{#1}{"0.03125_noisvarb0.75"}}{\(0.033\)}{\ifthenelse{\equal{#1}{0.03125_poissonic0.5}}{\(0.019\)}{\ifthenelse{\equal{#1}{0.03125}}{\(0.025\)}{\ifthenelse{\equal{#1}{0.0625}}{\(1961.978\)}{\ifthenelse{\equal{#1}{0.125}}{\(0.007\)}{\ifthenelse{\equal{#1}{"0.25_noisvarb0.75"}}{\(0.008\)}{\ifthenelse{\equal{#1}{0.25}}{\(0.006\)}{}}}}}}}}
\newcommand{\attentionGNMeanMSE}[1]{\ifthenelse{\equal{#1}{0.001953125}}{\(0.205\)}{\ifthenelse{\equal{#1}{0.00390625}}{\(0.163\)}{\ifthenelse{\equal{#1}{0.0078125}}{\(0.122\)}{\ifthenelse{\equal{#1}{0.015625}}{\(0.084\)}{\ifthenelse{\equal{#1}{"0.03125_noisvarb0.75"}}{\(0.056\)}{\ifthenelse{\equal{#1}{0.03125_poissonic0.5}}{\(0.043\)}{\ifthenelse{\equal{#1}{0.03125}}{\(0.053\)}{\ifthenelse{\equal{#1}{0.0625}}{\(0.032\)}{\ifthenelse{\equal{#1}{0.125}}{\(0.020\)}{\ifthenelse{\equal{#1}{"0.25_noisvarb0.75"}}{\(0.013\)}{\ifthenelse{\equal{#1}{0.25}}{\(0.012\)}{}}}}}}}}}}}}
\newcommand{\attentionGNStdMSE}[1]{\ifthenelse{\equal{#1}{0.001953125}}{\(0.122\)}{\ifthenelse{\equal{#1}{0.00390625}}{\(0.097\)}{\ifthenelse{\equal{#1}{0.0078125}}{\(0.067\)}{\ifthenelse{\equal{#1}{0.015625}}{\(0.043\)}{\ifthenelse{\equal{#1}{"0.03125_noisvarb0.75"}}{\(0.033\)}{\ifthenelse{\equal{#1}{0.03125_poissonic0.5}}{\(0.018\)}{\ifthenelse{\equal{#1}{0.03125}}{\(0.025\)}{\ifthenelse{\equal{#1}{0.0625}}{\(0.014\)}{\ifthenelse{\equal{#1}{0.125}}{\(0.007\)}{\ifthenelse{\equal{#1}{"0.25_noisvarb0.75"}}{\(0.005\)}{\ifthenelse{\equal{#1}{0.25}}{\(0.004\)}{}}}}}}}}}}}}
\newcommand{\attentionClustGNMeanMSE}[1]{\ifthenelse{\equal{#1}{0.0009765625}}{\(0.265\)}{\ifthenelse{\equal{#1}{0.001953125}}{\(0.234\)}{\ifthenelse{\equal{#1}{0.00390625}}{\(0.189\)}{\ifthenelse{\equal{#1}{0.0078125}}{\(0.140\)}{\ifthenelse{\equal{#1}{0.015625}}{\(0.096\)}{\ifthenelse{\equal{#1}{0.03125}}{\(0.061\)}{\ifthenelse{\equal{#1}{0.0625}}{\(0.036\)}{\ifthenelse{\equal{#1}{0.125}}{\(0.021\)}{\ifthenelse{\equal{#1}{0.25}}{\(0.013\)}{}}}}}}}}}}
\newcommand{\attentionClustGNStdMSE}[1]{\ifthenelse{\equal{#1}{0.0009765625}}{\(0.150\)}{\ifthenelse{\equal{#1}{0.001953125}}{\(0.140\)}{\ifthenelse{\equal{#1}{0.00390625}}{\(0.113\)}{\ifthenelse{\equal{#1}{0.0078125}}{\(0.081\)}{\ifthenelse{\equal{#1}{0.015625}}{\(0.051\)}{\ifthenelse{\equal{#1}{0.03125}}{\(0.029\)}{\ifthenelse{\equal{#1}{0.0625}}{\(0.016\)}{\ifthenelse{\equal{#1}{0.125}}{\(0.008\)}{\ifthenelse{\equal{#1}{0.25}}{\(0.004\)}{}}}}}}}}}}
\newcommand{\attentionGNVarbNoisMeanMSE}[1]{\ifthenelse{\equal{#1}{"0.03125_noisvarb0.75"}}{\(0.052\)}{\ifthenelse{\equal{#1}{0.03125_poissonic0.5}}{\(0.043\)}{\ifthenelse{\equal{#1}{0.03125}}{\(0.054\)}{\ifthenelse{\equal{#1}{"0.25_noisvarb0.75"}}{\(0.012\)}{\ifthenelse{\equal{#1}{0.25}}{\(0.012\)}{}}}}}}
\newcommand{\attentionGNVarbNoisStdMSE}[1]{\ifthenelse{\equal{#1}{"0.03125_noisvarb0.75"}}{\(0.031\)}{\ifthenelse{\equal{#1}{0.03125_poissonic0.5}}{\(0.019\)}{\ifthenelse{\equal{#1}{0.03125}}{\(0.025\)}{\ifthenelse{\equal{#1}{"0.25_noisvarb0.75"}}{\(0.005\)}{\ifthenelse{\equal{#1}{0.25}}{\(0.004\)}{}}}}}}
\newcommand{\attentionGNVarbNoisFixTrSNRMeanMSE}[1]{\ifthenelse{\equal{#1}{0.001953125}}{\(0.258\)}{\ifthenelse{\equal{#1}{0.00390625}}{\(0.175\)}{\ifthenelse{\equal{#1}{0.0078125}}{\(0.125\)}{\ifthenelse{\equal{#1}{0.015625}}{\(0.085\)}{\ifthenelse{\equal{#1}{0.03125}}{\(0.054\)}{\ifthenelse{\equal{#1}{0.0625}}{\(0.032\)}{\ifthenelse{\equal{#1}{0.125}}{\(0.020\)}{\ifthenelse{\equal{#1}{0.25}}{\(0.014\)}{}}}}}}}}}
\newcommand{\attentionGNVarbNoisFixTrSNRStdMSE}[1]{\ifthenelse{\equal{#1}{0.001953125}}{\(0.132\)}{\ifthenelse{\equal{#1}{0.00390625}}{\(0.100\)}{\ifthenelse{\equal{#1}{0.0078125}}{\(0.069\)}{\ifthenelse{\equal{#1}{0.015625}}{\(0.043\)}{\ifthenelse{\equal{#1}{0.03125}}{\(0.025\)}{\ifthenelse{\equal{#1}{0.0625}}{\(0.014\)}{\ifthenelse{\equal{#1}{0.125}}{\(0.007\)}{\ifthenelse{\equal{#1}{0.25}}{\(0.004\)}{}}}}}}}}}
\newcommand{\attentionGNVarbNoisFixTrLowSNRMeanMSE}[1]{\ifthenelse{\equal{#1}{0.0009765625}}{\(0.241\)}{\ifthenelse{\equal{#1}{0.001953125}}{\(0.204\)}{\ifthenelse{\equal{#1}{0.00390625}}{\(0.163\)}{\ifthenelse{\equal{#1}{0.0078125}}{\(0.123\)}{\ifthenelse{\equal{#1}{0.015625}}{\(0.087\)}{\ifthenelse{\equal{#1}{0.03125}}{\(0.059\)}{\ifthenelse{\equal{#1}{0.0625}}{\(0.044\)}{\ifthenelse{\equal{#1}{0.125}}{\(0.038\)}{\ifthenelse{\equal{#1}{0.25}}{\(0.035\)}{}}}}}}}}}}
\newcommand{\attentionGNVarbNoisFixTrLowSNRStdMSE}[1]{\ifthenelse{\equal{#1}{0.0009765625}}{\(0.140\)}{\ifthenelse{\equal{#1}{0.001953125}}{\(0.123\)}{\ifthenelse{\equal{#1}{0.00390625}}{\(0.097\)}{\ifthenelse{\equal{#1}{0.0078125}}{\(0.068\)}{\ifthenelse{\equal{#1}{0.015625}}{\(0.043\)}{\ifthenelse{\equal{#1}{0.03125}}{\(0.025\)}{\ifthenelse{\equal{#1}{0.0625}}{\(0.016\)}{\ifthenelse{\equal{#1}{0.125}}{\(0.013\)}{\ifthenelse{\equal{#1}{0.25}}{\(0.012\)}{}}}}}}}}}}
\newcommand{\attentionGNPoissMeanMSE}[1]{\ifthenelse{\equal{#1}{0.03125_poissonic0.5}}{\(0.038\)}{\ifthenelse{\equal{#1}{0.03125}}{\(0.065\)}{}}}
\newcommand{\attentionGNPoissStdMSE}[1]{\ifthenelse{\equal{#1}{0.03125_poissonic0.5}}{\(0.018\)}{\ifthenelse{\equal{#1}{0.03125}}{\(0.032\)}{}}}
\newcommand{\wienerMeanMSE}[1]{\ifthenelse{\equal{#1}{0.0009765625}}{\(0.287\)}{\ifthenelse{\equal{#1}{0.001953125}}{\(0.268\)}{\ifthenelse{\equal{#1}{0.00390625}}{\(0.240\)}{\ifthenelse{\equal{#1}{0.0078125}}{\(0.206\)}{\ifthenelse{\equal{#1}{0.015625}}{\(0.169\)}{\ifthenelse{\equal{#1}{0.03125}}{\(0.135\)}{\ifthenelse{\equal{#1}{0.0625}}{\(0.105\)}{\ifthenelse{\equal{#1}{0.125}}{\(0.080\)}{\ifthenelse{\equal{#1}{0.25}}{\(0.060\)}{}}}}}}}}}}
\newcommand{\wienerStdMSE}[1]{\ifthenelse{\equal{#1}{0.0009765625}}{\(0.147\)}{\ifthenelse{\equal{#1}{0.001953125}}{\(0.137\)}{\ifthenelse{\equal{#1}{0.00390625}}{\(0.120\)}{\ifthenelse{\equal{#1}{0.0078125}}{\(0.097\)}{\ifthenelse{\equal{#1}{0.015625}}{\(0.074\)}{\ifthenelse{\equal{#1}{0.03125}}{\(0.054\)}{\ifthenelse{\equal{#1}{0.0625}}{\(0.039\)}{\ifthenelse{\equal{#1}{0.125}}{\(0.028\)}{\ifthenelse{\equal{#1}{0.25}}{\(0.020\)}{}}}}}}}}}}
\newcommand{\polcnnMeanMSE}[1]{\ifthenelse{\equal{#1}{0.0009765625}}{\(0.271\)}{\ifthenelse{\equal{#1}{0.001953125}}{\(0.249\)}{\ifthenelse{\equal{#1}{0.00390625}}{\(0.209\)}{\ifthenelse{\equal{#1}{0.0078125}}{\(0.167\)}{\ifthenelse{\equal{#1}{0.015625}}{\(0.127\)}{\ifthenelse{\equal{#1}{0.03125}}{\(0.093\)}{\ifthenelse{\equal{#1}{0.0625}}{\(0.067\)}{\ifthenelse{\equal{#1}{0.125}}{\(0.047\)}{\ifthenelse{\equal{#1}{0.25}}{\(0.033\)}{}}}}}}}}}}
\newcommand{\polcnnStdMSE}[1]{\ifthenelse{\equal{#1}{0.0009765625}}{\(0.152\)}{\ifthenelse{\equal{#1}{0.001953125}}{\(0.142\)}{\ifthenelse{\equal{#1}{0.00390625}}{\(0.119\)}{\ifthenelse{\equal{#1}{0.0078125}}{\(0.089\)}{\ifthenelse{\equal{#1}{0.015625}}{\(0.062\)}{\ifthenelse{\equal{#1}{0.03125}}{\(0.040\)}{\ifthenelse{\equal{#1}{0.0625}}{\(0.026\)}{\ifthenelse{\equal{#1}{0.125}}{\(0.016\)}{\ifthenelse{\equal{#1}{0.25}}{\(0.011\)}{}}}}}}}}}}

\begin{figure*}
\begin{center}

\newcommand{\fnt}{\footnotesize \centering}
\newcommand{\modellabelwidth}{20mm}
\newcommand{\imagespacing}{1mm}
\newcommand{\imagesize}{20mm}
\newcommand{\imgone}{0007}
\newcommand{\imgtwo}{0011}
\newcommand{\imgthree}{0018}
\newcommand{\imgfour}{0026}
\newcommand{\imgfive}{0031}

\newcommand{\imgonepair}{0003-0008}
\newcommand{\imgtwopair}{0005-0008}
\newcommand{\imgthreepair}{0009-0000}
\newcommand{\imgfourpair}{0013-0000}
\newcommand{\imgfivepair}{0015-0008}

\newcommand{\standarddenoisproj}[3]{\includegraphics[width=\imagesize]{figures/denoiserExamples/snr#1/#2_#3.png}}
\newcommand{\ctfdenoisproj}[3]{\includegraphics[width=\imagesize]{figures/denoiserExamples/snr#1_ctf_shift8/#2_#3.png}}
\newcommand{\parakeetdenoisproj}[3]{\includegraphics[width=\imagesize]{figures/denoiserExamples/parakeet/#1/set#2_img#3.png}}
\newcommand{\thissnrstandardexample}[2]{\standarddenoisproj{0.03125}{#1}{#2}}
\newcommand{\thissnrctfexample}[2]{\ctfdenoisproj{0.25}{#1}{#2}}
\begin{tikzpicture}
    \node (clean-01) [inner sep=0mm] {\thissnrctfexample{clean}{\imgone}};
    \node (clean-02) [below=\imagespacing of clean-01, inner sep=0mm] {\thissnrctfexample{clean}{\imgtwo}};
    \node (clean-03) [below=\imagespacing of clean-02, inner sep=0mm] {\thissnrctfexample{clean}{\imgthree}};
    \node (clean-04) [below=\imagespacing of clean-03, inner sep=0mm] {\thissnrctfexample{clean}{\imgfour}};
    \node (clean-05) [below=\imagespacing of clean-04, inner sep=0mm] {\thissnrctfexample{clean}{\imgfive}};
    \node (clean-para) [below=\imagespacing of clean-05, inner sep=0mm] {\parakeetdenoisproj{clean}{4}{0}};


    \node (clean-label) [above=\imagespacing of clean-01, align=center, text width=\modellabelwidth] {\fnt Clean};

    \node (noisy-01) [right=\imagespacing of clean-01, inner sep=0mm] {\thissnrctfexample{noisy}{\imgone}};
    \node (noisy-02) [right=\imagespacing of clean-02, inner sep=0mm] {\thissnrctfexample{noisy}{\imgtwo}};
    \node (noisy-03) [right=\imagespacing of clean-03, inner sep=0mm] {\thissnrctfexample{noisy}{\imgthree}};
    \node (noisy-04) [right=\imagespacing of clean-04, inner sep=0mm] {\thissnrctfexample{noisy}{\imgfour}};
    \node (noisy-05) [right=\imagespacing of clean-05, inner sep=0mm] {\thissnrctfexample{noisy}{\imgfive}};
    \node (noisy-para) [right=\imagespacing of clean-para, inner sep=0mm] {\parakeetdenoisproj{noisy}{4}{0}};

    \node (noisy-label) [above=\imagespacing of noisy-01, align=center, text width=\modellabelwidth] {\fnt Noisy};

    \node (dncnn-01) [right=\imagespacing of noisy-01, inner sep=0mm] {\thissnrctfexample{dncnn}{\imgone}};
    \node (dncnn-02) [right=\imagespacing of noisy-02, inner sep=0mm] {\thissnrctfexample{dncnn}{\imgtwo}};
    \node (dncnn-03) [right=\imagespacing of noisy-03, inner sep=0mm] {\thissnrctfexample{dncnn}{\imgthree}};
    \node (dncnn-04) [right=\imagespacing of noisy-04, inner sep=0mm] {\thissnrctfexample{dncnn}{\imgfour}};
    \node (dncnn-05) [right=\imagespacing of noisy-05, inner sep=0mm] {\thissnrctfexample{dncnn}{\imgfive}};
    \node (dncnn-para) [right=\imagespacing of noisy-para, inner sep=0mm] {\parakeetdenoisproj{dncnn}{4}{0}};

    \node (dncnn-label) [above=\imagespacing of dncnn-01, align=center, text width=\modellabelwidth] {\fnt DnCNN};

    \node (dnunet-01) [right=\imagespacing of dncnn-01, inner sep=0mm] {\thissnrctfexample{dnunet}{\imgone}};
    \node (dnunet-02) [right=\imagespacing of dncnn-02, inner sep=0mm] {\thissnrctfexample{dnunet}{\imgtwo}};
    \node (dnunet-03) [right=\imagespacing of dncnn-03, inner sep=0mm] {\thissnrctfexample{dnunet}{\imgthree}};
    \node (dnunet-04) [right=\imagespacing of dncnn-04, inner sep=0mm] {\thissnrctfexample{dnunet}{\imgfour}};
    \node (dnunet-05) [right=\imagespacing of dncnn-05, inner sep=0mm] {\thissnrctfexample{dnunet}{\imgfive}};
    \node (dnunet-para) [right=\imagespacing of dncnn-para, inner sep=0mm] {\parakeetdenoisproj{unet}{4}{0}};

    \node (dnunet-label) [above=\imagespacing of dnunet-01, align=center, text width=\modellabelwidth] {\fnt U-Net};

    \node (polcnn-01) [right=\imagespacing of dnunet-01, inner sep=0mm] {\thissnrctfexample{polcnn}{\imgone}};
    \node (polcnn-02) [right=\imagespacing of dnunet-02, inner sep=0mm] {\thissnrctfexample{polcnn}{\imgtwo}};
    \node (polcnn-03) [right=\imagespacing of dnunet-03, inner sep=0mm] {\thissnrctfexample{polcnn}{\imgthree}};
    \node (polcnn-04) [right=\imagespacing of dnunet-04, inner sep=0mm] {\thissnrctfexample{polcnn}{\imgfour}};
    \node (polcnn-05) [right=\imagespacing of dnunet-05, inner sep=0mm] {\thissnrctfexample{polcnn}{\imgfive}};
    \node (polcnn-para) [right=\imagespacing of dnunet-para, inner sep=0mm] {\parakeetdenoisproj{polcnn}{4}{0}};

    \node (polcnn-label) [above=\imagespacing of polcnn-01, align=center, text width=\modellabelwidth] {\fnt Polar CNN};

    \node (trafo-dir-01) [right=\imagespacing of polcnn-01, inner sep=0mm] {\thissnrctfexample{ss8}{\imgone}};
    \node (trafo-dir-02) [right=\imagespacing of polcnn-02, inner sep=0mm] {\thissnrctfexample{ss8}{\imgtwo}};
    \node (trafo-dir-03) [right=\imagespacing of polcnn-03, inner sep=0mm] {\thissnrctfexample{ss8}{\imgthree}};
    \node (trafo-dir-04) [right=\imagespacing of polcnn-04, inner sep=0mm] {\thissnrctfexample{ss8}{\imgfour}};
    \node (trafo-dir-05) [right=\imagespacing of polcnn-05, inner sep=0mm] {\thissnrctfexample{ss8}{\imgfive}};
    \node (trafo-dir-para) [right=\imagespacing of polcnn-para, inner sep=0mm] {\parakeetdenoisproj{ss8_grpNrm4}{0}{4}};

    \node (trafo-dir-label) [above=\imagespacing of trafo-dir-01, align=center, text width=\modellabelwidth] {\fnt Polar transf. (dir.)};

    \node (trafo-gen-01) [right=\imagespacing of trafo-dir-01, inner sep=0mm] {\thissnrctfexample{ss16}{\imgonepair}};
    \node (trafo-gen-02) [right=\imagespacing of trafo-dir-02, inner sep=0mm] {\thissnrctfexample{ss16}{\imgtwopair}};
    \node (trafo-gen-03) [right=\imagespacing of trafo-dir-03, inner sep=0mm] {\thissnrctfexample{ss16}{\imgthreepair}};
    \node (trafo-gen-04) [right=\imagespacing of trafo-dir-04, inner sep=0mm] {\thissnrctfexample{ss16}{\imgfourpair}};
    \node (trafo-gen-05) [right=\imagespacing of trafo-dir-05, inner sep=0mm] {\thissnrctfexample{ss16}{\imgfivepair}};
    \node (trafo-gen-para) [right=\imagespacing of trafo-dir-para, inner sep=0mm] {\parakeetdenoisproj{ss16_grpNrm4_dir2}{0}{4}};

    \node (trafo-gen-label) [above=\imagespacing of trafo-gen-01, align=center, text width=\modellabelwidth] {\fnt Polar transf. (gen.)};
\end{tikzpicture}
\hfill
\end{center}
\caption{\label{fig:denoisedExamples} Sample images with CTF, shifts, and Gaussian noise at $\text{SNR} = 0.02$ (top five rows), and for the Parakeet simulator (bottom row).}
\end{figure*}

\begin{figure*}
\begin{center}
\begin{subfigure}[c]{0.70\textwidth}
 \includegraphics[width=0.32\textwidth]{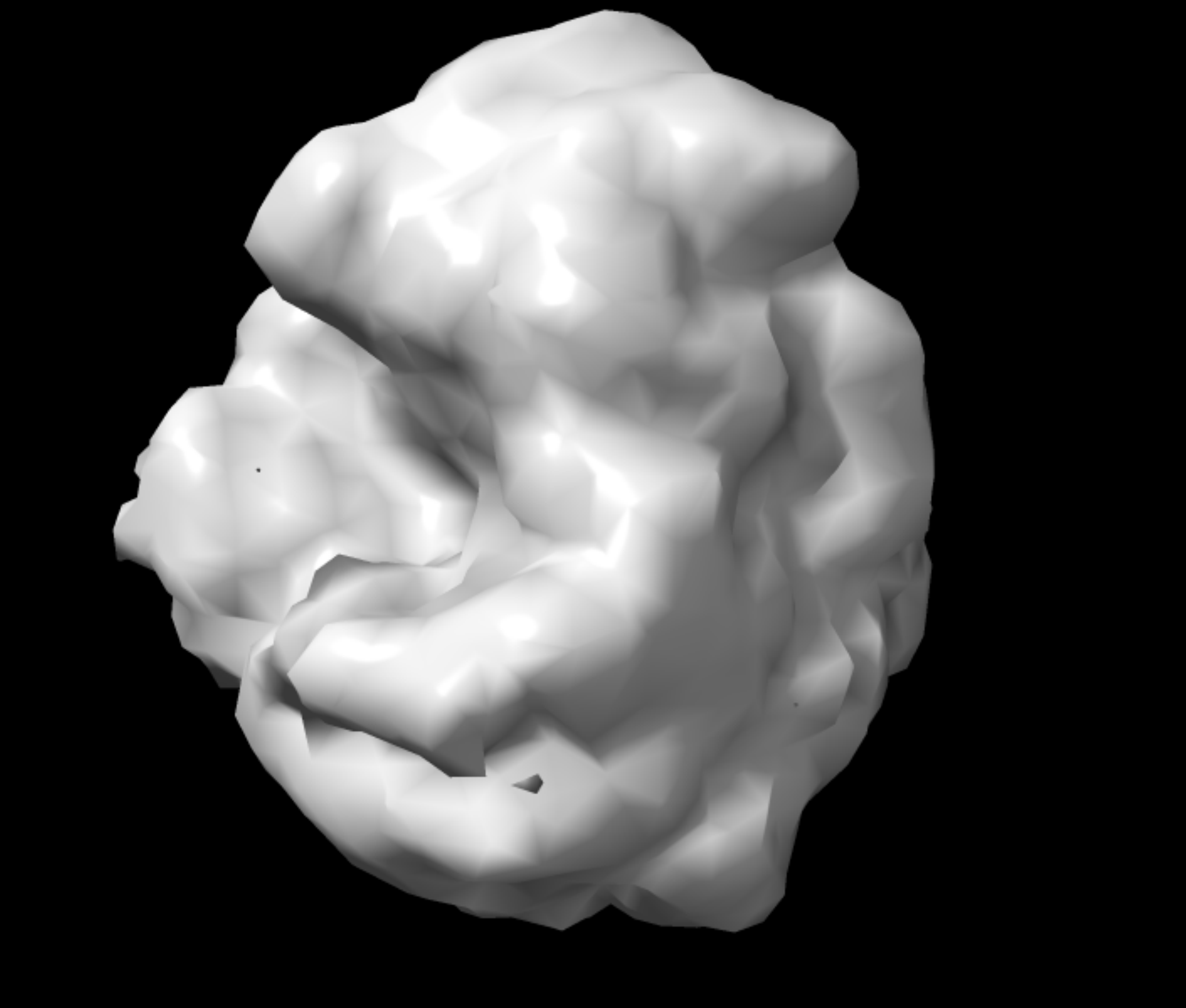}
 \includegraphics[width=0.32\textwidth]{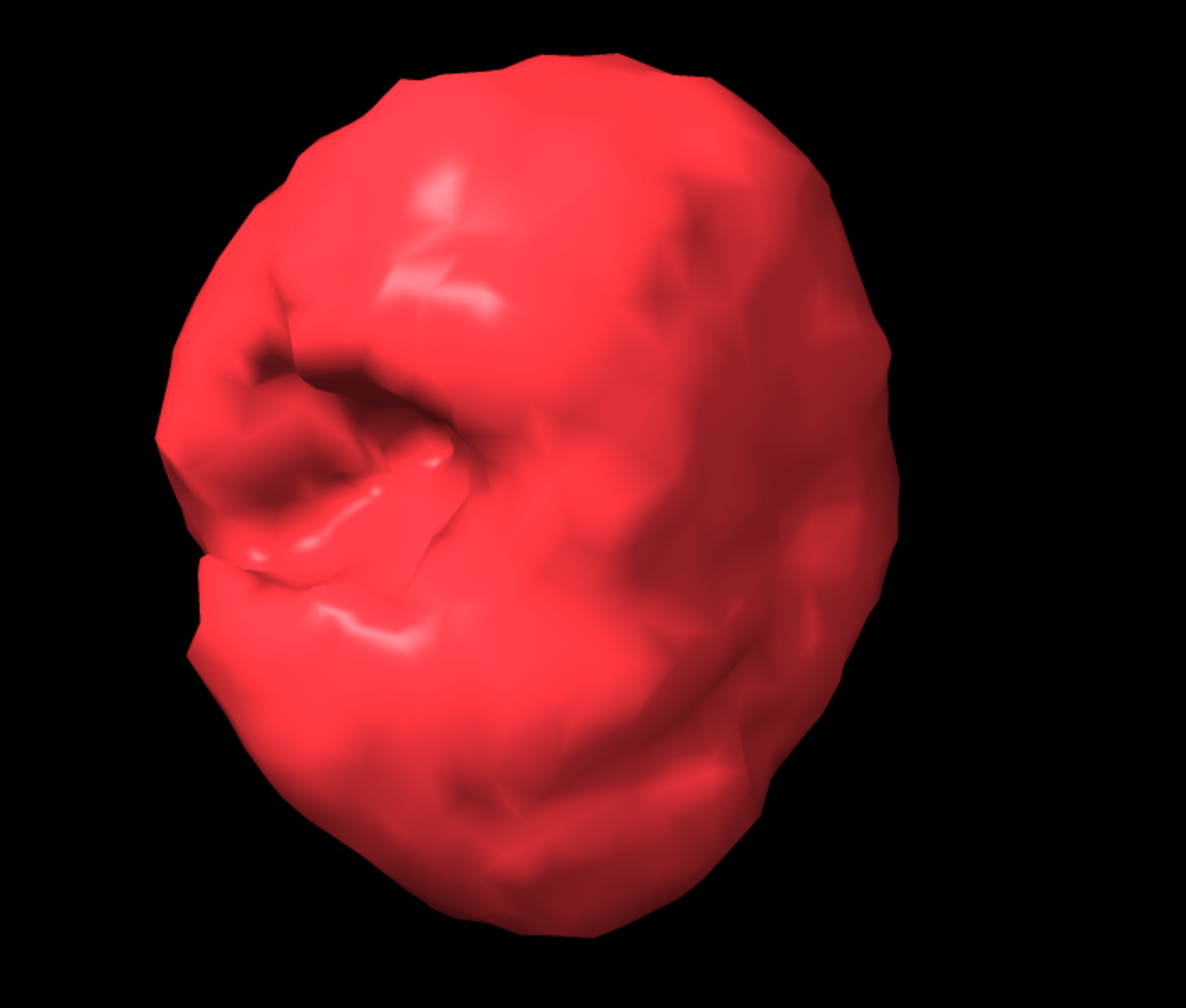}
 \includegraphics[width=0.32\textwidth]{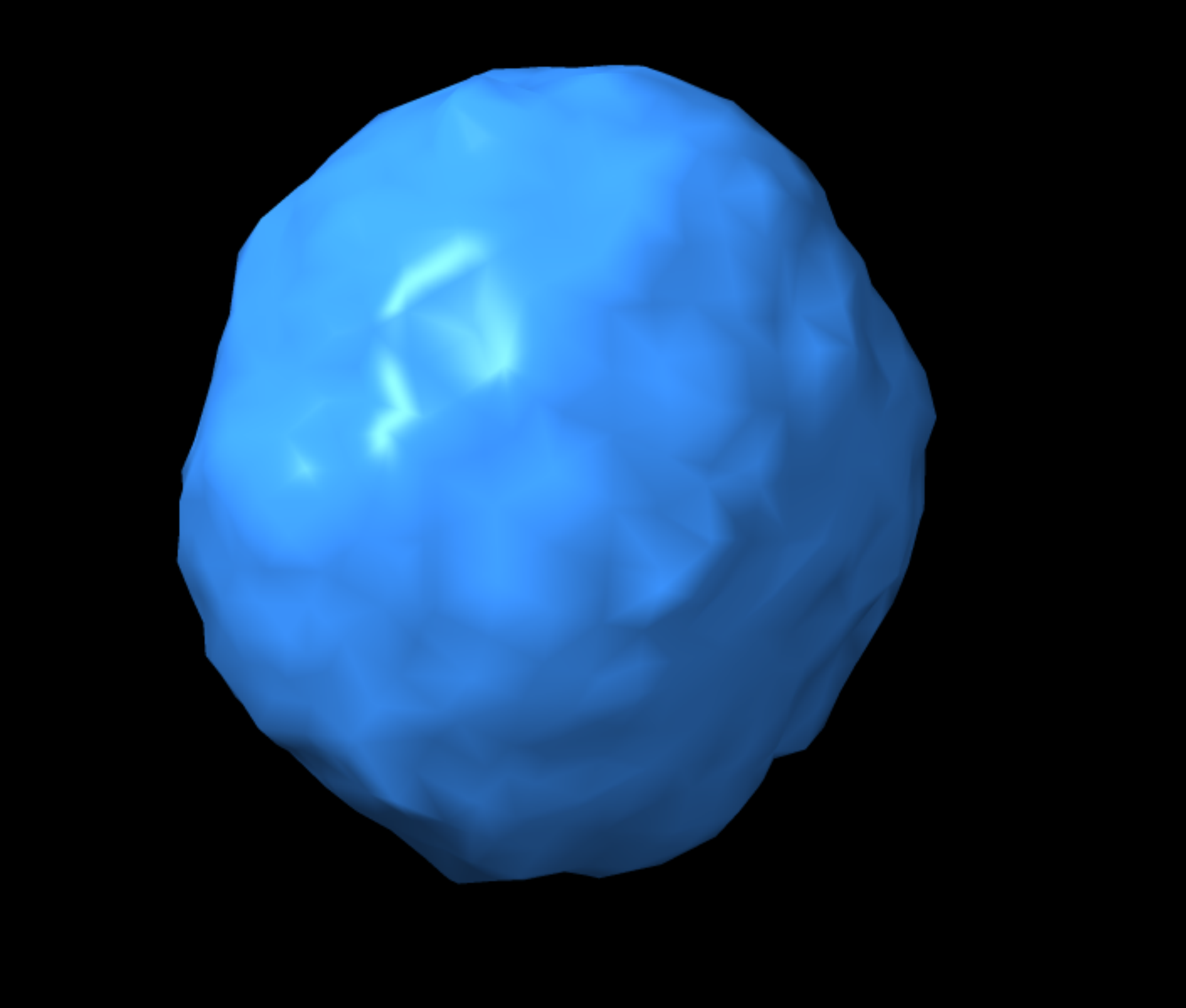}
\end{subfigure}
\begin{subfigure}[c]{0.29\textwidth}
    \begin{tikzpicture}
        \node (fsc) {\includegraphics[width=50mm, trim={3mm 2mm 0mm 0mm}, clip]{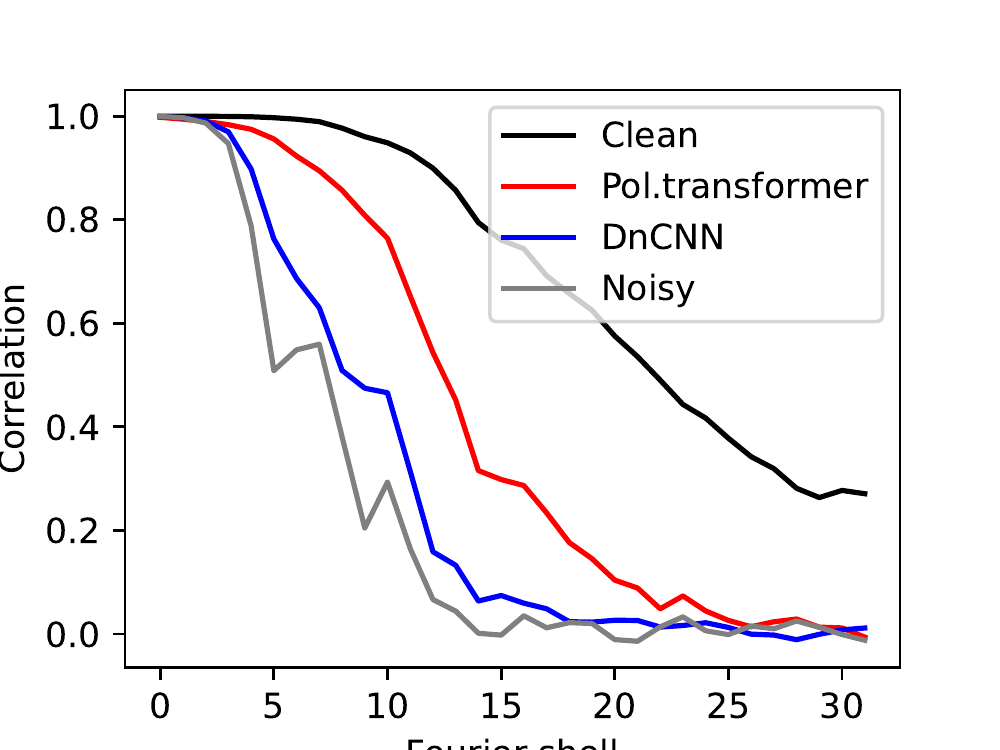}};
        \node [below of=fsc, yshift=-12mm] {\footnotesize Fourier shell};
        \node [left of=fsc, xshift=-17mm, rotate=90] {\footnotesize Correlation};
    \end{tikzpicture}
\end{subfigure}
\end{center}
    \caption{(left) 3D reconstructions of DMSO reductase from denoised images: ground truth (gray), polar transformer (red), and DnCNN (blue). (right) FSC curves of the reconstructions.}
 \label{fig:reconstrFSC}
\end{figure*}

\subsection{Denoising Results}

Table~\ref{table:results} shows the relative MSE ($\| f(y) - x\|^2 / \|x\|^2$) results on the denoising task for $\text{SNR} = 0.02$ (measured as the average pixel energy of the clean images divided by the noise variance).
First, we see that the polar CNN outperforms the other single-image denoisers (the DnCNN and U-Net) which do not explicitly encode rotational equivariance (and therefore perform slightly worse) but learn it through the data augmentation implicit in the training data.
Going beyond the polar CNN, we see that the polar transformer model consistently outperforms single-image methods when given a directional image set (i.e., the underlying clean images are the same up to in-plane rotation).
By studying the attention coefficients, we can also verify that the transformer is able to correctly recover the rotation angle in a robust manner (see Appendix~\ref{sec:angularAttention}).
Finally, we see that the polar transformer is also able to cluster sets of noisy images in the general set denoising task.
While performance is strictly worse, it remains quite close to the directional set results.
In terms of inference time, the single-image denoisers take on the order of $0.1~\text{ms}$ per image while the polar transformer (for sets of eight images) takes on the order of $1~\text{ms}$ per image on an A100 GPU.

We should also note that all methods outperform corresponding class averaging approaches.
For an $\text{SNR}$ of $0.02$, the relative MSE of the noisy image is $1/0.02 = 50$.
Since class averaging reduces noise by averaging a set of $K$ images, the lowest relative MSE it can achieve is $50 / K$;
to reach the performance levels given in Table~\ref{table:results}, it would therefore need around 1000 images.

Looking at the activation coefficients of the attention module (Appendix~\ref{sec:clusteringAttention}), we see that these indeed cluster the different viewing directions as expected.
The example images of Figure~\ref{fig:denoisedExamples} (top five rows) show how the polar CNN produces comparable quality to that of DnCNN and U-Net, but that the polar transformer produces sharper, higher-quality images.
For more details, including PSNR, SSIM and resolution metrics as well as results for larger image sets, see Appendix~\ref{sec:moreResults}.

Finally, we note that the models are not very sensitive to the noise distribution: models trained on Gaussian noise can also be applied successfully to Poisson-noise images (see Appendix~\ref{sec:poissonNoiseDetails}).
There is a small loss in MSE, but overall features are preserved.

\begin{table}
    \begin{center}
        \begin{tabular}{l|c}
            Method & Relative MSE \\
            \hline
            DnCNN & 0.088 \\
            U-Net & 0.089 \\
            Polar CNN & 0.081 \\ 
            \hline
            Polar transformer (dir.) & 0.042 \\ 
            Polar transformer (gen.) & 0.049 \\ 
        \end{tabular}
    \end{center}
    \caption{
        \label{table:results}
        Denoising performance of the proposed methods compared to standard baselines on simulated cryo-EM data with CTFs and shifts for $\text{SNR} = 0.02$.
        The top three methods operate on single images while the transformer methods operate on sets of images (the directional regime has eight images, all in the same viewing direction, while the general regime has sixteen images split between two viewing directions).
    }
\end{table}

Although the forward model used in the above experiments incorporates the main aspects of the cryo-EM imaging process, it is an approximation.
To assess qualitatively how far the current models are from being usable with true micrographs, we also applied them to some projections generated with the Parakeet digital twin~\cite{Parkhurst2021Parakeet}, which uses multislice simulations to produce more accurate synthetic images.
We used Parakeet to generate a ``clean'' projection (by taking a very high electron dose) and a noisy one to which the single-image denoisers (DnCNN, U-Net, and polar CNN) were applied.
Each was able to recover the clean image well, but all suffered some loss of detail (see Figure~\ref{fig:denoisedExamples}, bottom row).

To evaluate the polar transformer, we generated an additional set of 10 000 projections and picked two sets of images of similar viewing direction (within ten degrees, but with arbitrary in-plane rotation).
The first set, containing eight images, was given to the directional-set transformer, while another, containing sixteen images, was given to the general-set transformer.
This was sufficient to yield improved results compared to the single-image denoisers (see Figure~\ref{fig:denoisedExamples}, bottom row).

\subsection{3D Reconstruction}
\label{sec:reconstruction}

While there is value in denoising projection images themselves, the ultimate goal of cryo-EM processing is to construct a 3D model of the imaged molecule.
This is typically done in two steps: an initial, low-resolution reconstruction, known as an \emph{ab initio} model~\citep{Punjani2017CryoSPARC, VanHeel1987Angular}, which is followed by high-resolution refinement.

We will focus on a class of methods based on common lines~\citep{Singer2011CommonLines, Caspy2021Cryo}, which relies on the intersection geometry of the projection images in Fourier space.
To evaluate the effect of denoising on these reconstruction methods, we generated $2 048$ noisy projection images of DMSO reductase (PDB ID 1DMR) at $\text{SNR} = 0.25$ that we processed either using a DnCNN or the polar transformer (trained on directional sets).
For the latter, we clustered the images into sets of eight images each according to their viewing direction (up to in-plane rotation), mimicking the 2D classification step.
These sets were then fed into the polar transformer model.

Both sets of denoised images were then used as input to the \texttt{sync3n} ab initio reconstruction algorithm implemented in the ASPIRE package~\cite{Aspire2025}.
Compared to the ground truth orientations, the DnCNN images achieved a mean angular distance of $111^\circ$, while the polar transformer achieved $69.7^\circ$.
The resulting orientations were then paired with the original noisy images (to avoid bias induced by the denoising process) to obtain 3D reconstructions by inverting the forward model using the method of least squares.
The results are shown in Figure~\ref{fig:reconstrFSC}.
We see that the increased accuracy in orientation estimation achieved by the polar transformer results in more accurate 3D models.
This is also reflected in the Fourier shell correlation~\citep{Harauz1986filters} (see Figure~\ref{fig:reconstrFSC}), where the estimated orientations from the polar transformer better capture higher-resolution information in the 3D density map.
For more reconstruction examples, see Appendix~\ref{sec:reconstrReliability}.

\section{CONCLUSION}
\label{sec:conclusion}

In this work, we have presented a new, powerful architecture for cryo-EM image processing: the polar transformer.
While this model has significant potential as demonstrated by its performance on the denoising task, more work remains before it can be applied to practical problems.
Nonetheless, we believe that the polar transformer provides an useful component to an end-to-end cryo-EM reconstruction pipeline.

\section*{Acknowledgments}
The authors would like to thank Eftychios Pnevmatikakis, who worked on an earlier prototype of the above architecture and provided invaluable advice.
Similarly, the authors would like to thank Olivier Verdier for his thoughtful comments in discussions of our approach and results.
Further thanks to Axel Janson for generating the Parakeet test data.
This work was funded by grants from the Göran Gustafsson Foundation and the Swedish Research Council through grant no. 2023-04143.
The computations were enabled by resources provided by the National Academic Infrastructure for Supercomputing in Sweden (NAISS), partially funded by the Swedish Research Council through grant agreement no. 2022-06725.
The Flatiron Institute is a division of the Simons Foundation.

\bibliography{aistats2026_paper}

\begin{thebibliography}{}

\bibitem[Barnett et~al., 2017]{Barnett2017Marching}
Barnett, A., Greengard, L., Pataki, A., and Spivak, M. (2017).
\newblock Rapid solution of the cryo-{EM} reconstruction problem by frequency
  marching.
\newblock {\em SIAM Journal on Imaging Sciences}, 10(3):1170--1195.

\bibitem[Barnett et~al., 2019]{Barnett2019FINUFFT}
Barnett, A.~H., Magland, J., and af~Klinteberg, L. (2019).
\newblock A parallel nonuniform fast {Fourier} transform library based on an
  ``exponential of semicircle'' kernel.
\newblock {\em SIAM Journal on Scientific Computing}, 41(5):C479--C504.

\bibitem[Bartesaghi et~al., 2018]{Bartesaghi2018Atomic}
Bartesaghi, A., Aguerrebere, C., Falconieri, V., Banerjee, S., Earl, L.~A.,
  Zhu, X., Grigorieff, N., Milne, J.~L., Sapiro, G., Wu, X., and Subramaniam,
  S. (2018).
\newblock Atomic resolution cryo-{EM} structure of $\beta$-galactosidase.
\newblock {\em Structure}, 26(6):848--856.e3.

\bibitem[Bendory et~al., 2020]{Bendory2020SPCEM}
Bendory, T., Bartesaghi, A., and Singer, A. (2020).
\newblock Single-particle cryo-electron microscopy: Mathematical theory,
  computational challenges, and opportunities.
\newblock {\em IEEE Signal Processing Magazine}, 37(2):58–76.

\bibitem[Bepler et~al., 2020]{Bepler2019TopazDenoise}
Bepler, T., Noble, A.~J., and Berger, B. (2020).
\newblock Topaz-denoise: general deep denoising models for {cryoEM} and
  {cryoET}.
\newblock {\em Nature Communications}, 11(1).

\bibitem[Bertero et~al., 2021]{Bertero2021InverseProblems}
Bertero, M., Boccacci, P., and De~Mol, C. (2021).
\newblock {\em Introduction to Inverse Problems in Imaging}.
\newblock 2 edition.

\bibitem[Bhamre et~al., 2016]{Bhamre2016Denois}
Bhamre, T., Zhang, T., and Singer, A. (2016).
\newblock Denoising and covariance estimation of single particle cryo-{EM}
  images.
\newblock {\em Journal of Structural Biology}, 195(1):72–81.

\bibitem[Bibas et~al., 2021]{Bibas2021RotInvFeat}
Bibas, K., Weiss-Dicker, G., Cohen, D., Cahan, N., and Greenspan, H. (2021).
\newblock Learning rotation invariant features for cryogenic electron
  microscopy image reconstruction.
\newblock In {\em Proc. ISBI}, pages 563--566.

\bibitem[Brown et~al., 2020]{Brown2020GPT3}
Brown, T.~B., Mann, B., Ryder, N., Subbiah, M., Kaplan, J., Dhariwal, P.,
  Neelakantan, A., Shyam, P., Sastry, G., Askell, A., Agarwal, S.,
  Herbert-Voss, A., Krueger, G., Henighan, T., Child, R., Ramesh, A., Ziegler,
  D.~M., Wu, J., Winter, C., Hesse, C., Chen, M., Sigler, E., Litwin, M., Gray,
  S., Chess, B., Clark, J., Berner, C., McCandlish, S., Radford, A., Sutskever,
  I., and Amodei, D. (2020).
\newblock Language models are few-shot learners.

\bibitem[Buchholz et~al., 2019]{Buchholz2019CryoCARE}
Buchholz, T.-O., Jordan, M., Pigino, G., and Jug, F. (2019).
\newblock Cryo-care: Content-aware image restoration for cryo-transmission
  electron microscopy data.
\newblock In {\em 2019 IEEE 16th International Symposium on Biomedical Imaging
  (ISBI 2019)}, pages 502--506.

\bibitem[Caspy et~al., 2021]{Caspy2021Cryo}
Caspy, I., Neumann, E., Fadeeva, M., Liveanu, V., Savitsky, A., Frank, A.,
  Kalisman, Y.~L., Shkolnisky, Y., Murik, O., Treves, H., et~al. (2021).
\newblock {Cryo-EM} photosystem i structure reveals adaptation mechanisms to
  extreme high light in {Chlorella ohadii}.
\newblock {\em Nature Plants}, 7(9):1314--1322.

\bibitem[Cohen et~al., 2018]{Cohen2018Spherical}
Cohen, T.~S., Geiger, M., Köhler, J., and Welling, M. (2018).
\newblock Spherical {CNN}s.
\newblock In {\em International Conference on Learning Representations}.

\bibitem[Devlin et~al., 2018]{Devlin2018Bert}
Devlin, J., Chang, M.-W., Lee, K., and Toutanova, K. (2018).
\newblock Bert: {Pre-training} of deep bidirectional transformers for language
  understanding.
\newblock {\em arXiv preprint}.

\bibitem[{\relax DLMF}, 2025]{DLMF}
{\relax DLMF} (2025).
\newblock {\it NIST Digital Library of Mathematical Functions}.
\newblock F.~W.~J. Olver, A.~B. {Olde Daalhuis}, D.~W. Lozier, B.~I. Schneider,
  R.~F. Boisvert, C.~W. Clark, B.~R. Miller, B.~V. Saunders, H.~S. Cohl, and
  M.~A. McClain, eds.

\bibitem[Dosovitskiy et~al., 2021]{Dosovitskiy2020VIT}
Dosovitskiy, A., Beyer, L., Kolesnikov, A., Weissenborn, D., Zhai, X.,
  Unterthiner, T., Dehghani, M., Minderer, M., Heigold, G., Gelly, S., et~al.
  (2021).
\newblock An image is worth 16x16 words: Transformers for image recognition at
  scale.
\newblock In {\em iclr}.

\bibitem[Dubochet et~al., 1982]{Dubochet1982CryoEM}
Dubochet, J., Lepault, J., Freeman, R., Berriman, J.~A., and Homo, J.-C.
  (1982).
\newblock Electron microscopy of frozen water and aqueous solutions.
\newblock {\em Journal of Microscopy}, 128(3):219--237.

\bibitem[Frank, 1996]{Frank1996CryoEM}
Frank, J. (1996).
\newblock {\em Three-Dimensional Electron Microscopy of Macromolecular
  Assemblies}.
\newblock Academic Press.

\bibitem[Frank et~al., 1996]{Frank1996SPIDER}
Frank, J., Radermacher, M., Penczek, P., Zhu, J., Li, Y., Ladjadj, M., and
  Leith, A. (1996).
\newblock {SPIDER and WEB}: processing and visualization of images in {3D}
  electron microscopy and related fields.
\newblock {\em Journal of Structural Biology}, 116(1):190--199.

\bibitem[Greengard and Lee, 2004]{Greengard2004NUFFT}
Greengard, L. and Lee, J.-Y. (2004).
\newblock Accelerating the nonuniform fast {Fourier} transform.
\newblock {\em SIAM Review}, 46(3):443--454.

\bibitem[Gupta et~al., 2021]{Gupta2021CryoGAN}
Gupta, H., McCann, M.~T., Donati, L., and Unser, M. (2021).
\newblock {CryoGAN}: {A} new reconstruction paradigm for single-particle
  cryo-{EM} via deep adversarial learning.
\newblock {\em IEEE Transactions on Computational Imaging}, 7:759--774.

\bibitem[Gurrola-Ramos et~al., 2021]{GurollaRamos2021Unet}
Gurrola-Ramos, J., Dalmau, O., and Alarc{\'o}n, T.~E. (2021).
\newblock A residual dense u-net neural network for image denoising.
\newblock {\em IEEE Access}, 9:31742--31754.

\bibitem[Harauz and van Heel, 1986]{Harauz1986filters}
Harauz, G. and van Heel, M. (1986).
\newblock Exact filters for general geometry three dimensional reconstruction.
\newblock {\em Optik}, 73(4):146--156.

\bibitem[Iudin et~al., 2022]{Iudin2022EMPIAR}
Iudin, A., Korir, P.~K., Somasundharam, S., Weyand, S., Cattavitello, C.,
  Fonseca, N., Salih, O., Kleywegt, G., and Patwardhan, A. (2022).
\newblock {EMPIAR}: the electron microscopy public image archive.
\newblock {\em Nucleic Acids Research}, 51(D1):D1503--D1511.

\bibitem[Kay, 1993]{Kay1993Fundamentals}
Kay, S.~M. (1993).
\newblock {\em Fundamentals of Statistical Signal Processing: Estimation
  Theory}.
\newblock Prentice Hall.

\bibitem[Kimanius et~al., 2024]{Kimanius2024Prior}
Kimanius, D., Jamali, K., Wilkinson, M.~E., L{\"o}vestam, S., Velazhahan, V.,
  Nakane, T., and Scheres, S.~H. (2024).
\newblock Data-driven regularization lowers the size barrier of {cryo-EM}
  structure determination.
\newblock {\em Nature Methods}, 21(7):1216--1221.

\bibitem[Kimanius et~al., 2021]{Kimanius2021Prior}
Kimanius, D., Zickert, G., Nakane, T., Adler, J., Lunz, S., Sch{\"o}nlieb,
  C.-B., {\"O}ktem, O., and Scheres, S.~H. (2021).
\newblock Exploiting prior knowledge about biological macromolecules in
  {cryo-EM} structure determination.
\newblock {\em IUCrJ}, 8:60--75.

\bibitem[Kingma and Ba, 2017]{Kingma2017Adam}
Kingma, D.~P. and Ba, J. (2017).
\newblock Adam: {A} method for stochastic optimization.
\newblock {\em arXiv preprint}.

\bibitem[Kondor et~al., 2018]{Kondor2018CGNets}
Kondor, R., Lin, Z., and Trivedi, S. (2018).
\newblock {Clebsch--Gordan} nets: a fully {Fourier} space spherical
  convolutional neural network.
\newblock In {\em Advances in Neural Information Processing Systems},
  volume~31.

\bibitem[Kondor and Trivedi, 2018]{Kondor2018Equivariance}
Kondor, R. and Trivedi, S. (2018).
\newblock On the generalization of equivariance and convolution in neural
  networks to the action of compact groups.
\newblock In {\em Proceedings of the 35th International Conference on Machine
  Learning}, volume~80, pages 2747--2755. PMLR.

\bibitem[Kwon et~al., 2023]{Kwon2023RotInvRepr}
Kwon, S., Choi, J.~Y., and Ryu, E.~K. (2023).
\newblock Rotation and translation invariant representation learning with
  implicit neural representations.
\newblock In {\em Proc. ICML}, volume 202, pages 18037--18056. PMLR.

\bibitem[Lehtinen et~al., 2018]{Lehtinen2018Noise2Noise}
Lehtinen, J., Munkberg, J., Hasselgren, J., Laine, S., Karras, T., Aittala, M.,
  and Aila, T. (2018).
\newblock {N}oise2{N}oise: {Learning} image restoration without clean data.
\newblock In {\em Proc. ICML}, volume~80, pages 2965--2974.

\bibitem[Levy et~al., 2022]{Levy2022CryoFIRE}
Levy, A., Wetzstein, G., Martel, J.~N., Poitevin, F., and Zhong, E. (2022).
\newblock Amortized inference for heterogeneous reconstruction in cryo-{EM}.
\newblock In {\em Proc. NeurIPS}, volume~35, pages 13038--13049.

\bibitem[Nashed et~al., 2021]{Nashed2021CryoPoseNet}
Nashed, Y.~S., Poitevin, F., Gupta, H., Woollard, G., Kagan, M., Yoon, C.~H.,
  and Ratner, D. (2021).
\newblock {CryoPoseNet}: {End-to-end} simultaneous learning of single-particle
  orientation and {3D} map reconstruction from cryo-electron microscopy data.
\newblock In {\em Proc. ICCV}, pages 4066--4076.

\bibitem[Nasiri and Bepler, 2022]{Nasiri2022VAE}
Nasiri, A. and Bepler, T. (2022).
\newblock Unsupervised object representation learning using translation and
  rotation group equivariant {VAE}.
\newblock In {\em Proc. NeurIPS}, volume~35, pages 15255--15267.

\bibitem[Palovcak et~al., 2020]{Palovcak2020Denoising}
Palovcak, E., Asarnow, D., Campbell, M.~G., Yu, Z., and Cheng, Y. (2020).
\newblock {Enhancing the signal-to-noise ratio and generating contrast for
  {cryo-EM} images with convolutional neural networks}.
\newblock {\em IUCrJ}, 7(6):1142--1150.

\bibitem[Park and Chirikjian, 2014]{Park2014ClassAvg}
Park, W. and Chirikjian, G.~S. (2014).
\newblock An assembly automation approach to alignment of noncircular
  projections in electron microscopy.
\newblock {\em IEEE Transactions on Automation Science and Engineering},
  11(3):668--679.

\bibitem[Parkhurst et~al., 2021]{Parkhurst2021Parakeet}
Parkhurst, J.~M., Dumoux, M., Basham, M., Clare, D., Siebert, C.~A., Varslot,
  T., Kirkland, A., Naismith, J.~H., and Evans, G. (2021).
\newblock Parakeet: a digital twin software pipeline to assess the impact of
  experimental parameters on tomographic reconstructions for cryo-electron
  tomography.
\newblock {\em Open Biology}, 11(10):210160.

\bibitem[Punjani et~al., 2017]{Punjani2017CryoSPARC}
Punjani, A., Rubinstein, J.~L., Fleet, D.~J., and Brubaker, M.~A. (2017).
\newblock {cryoSPARC}: algorithms for rapid unsupervised cryo-{EM} structure
  determination.
\newblock {\em Nature Methods}, 14(3):290--296.

\bibitem[Ralston and Rabinowitz, 2001]{Ralston2001NumericalAnalysis}
Ralston, A. and Rabinowitz, P. (2001).
\newblock {\em A first course in numerical analysis}.
\newblock Courier Corporation.

\bibitem[Ronneberger et~al., 2015]{Ronneberger2015Unet}
Ronneberger, O., Fischer, P., and Brox, T. (2015).
\newblock U-net: {Convolutional} networks for biomedical image segmentation.
\newblock In {\em MICCAI 2015}, pages 234--241. Springer.

\bibitem[Scheres, 2012]{Scheres2012Relion}
Scheres, S.~H. (2012).
\newblock {RELION}: {Implementation} of a {Bayesian} approach to cryo-{EM}
  structure determination.
\newblock {\em Journal of Structural Biology}, 180(3):519--530.

\bibitem[Scheres, 2015]{Scheres2015ClassAvg}
Scheres, S.~H. (2015).
\newblock Semi-automated selection of cryo-{EM} particles in {RELION}-1.3.
\newblock {\em Journal of Structural Biology}, 189(2):114--122.

\bibitem[Schwab et~al., 2024]{Schwab2024DynaMight}
Schwab, J., Kimanius, D., Burt, A., Dendooven, T., and Scheres, S.~H. (2024).
\newblock {DynaMight}: estimating molecular motions with improved
  reconstruction from {cryo-EM} images.
\newblock {\em Nature Methods}, pages 1--8.

\bibitem[Shih et~al., 2021]{Shih2021cuFINUFFT}
Shih, Y.-h., Wright, G., Andén, J., Blaschke, J., and Barnett, A.~H. (2021).
\newblock {cuFINUFFT}: a load-balanced {GPU} library for general-purpose
  nonuniform {FFTs}.
\newblock In {\em 2021 IEEE International Parallel and Distributed Processing
  Symposium Workshops (IPDPSW)}, pages 688--697.

\bibitem[Sindelar and Grigorieff, 2011]{Sindelar2011Wiener}
Sindelar, C.~V. and Grigorieff, N. (2011).
\newblock An adaptation of the {Wiener} filter suitable for analyzing images of
  isolated single particles.
\newblock {\em Journal of structural biology}, 176(1):60--74.

\bibitem[Singer and Shkolnisky, 2011]{Singer2011CommonLines}
Singer, A. and Shkolnisky, Y. (2011).
\newblock Three-dimensional structure determination from common lines in
  {Cryo-EM} by eigenvectors and semidefinite programming.
\newblock {\em SIAM Journal on Imaging Sciences}, 4(2):543--572.

\bibitem[Singer and Sigworth, 2020]{Singer2020Methods}
Singer, A. and Sigworth, F.~J. (2020).
\newblock Computational methods for single-particle electron cryomicroscopy.
\newblock {\em Annual Review of Biomedical Data Science}, 3(Volume 3,
  2020):163--190.

\bibitem[Stein and Weiss, 1971]{Stein1971Fourier}
Stein, E.~M. and Weiss, G. (1971).
\newblock {\em Introduction to Fourier Analysis on Euclidean Spaces}.
\newblock Princeton University Press.

\bibitem[Tang et~al., 2007]{Tang2007Eman2}
Tang, G., Peng, L., Baldwin, P.~R., Mann, D.~S., Jiang, W., Rees, I., and
  Ludtke, S.~J. (2007).
\newblock {EMAN2}: {An} extensible image processing suite for electron
  microscopy.
\newblock {\em Journal of Structural Biology}, 157(1):38--46.
\newblock Software tools for macromolecular microscopy.

\bibitem[van Heel, 1984]{Heel1984MVClassif}
van Heel, M. (1984).
\newblock Multivariate statistical classification of noisy images (randomly
  oriented biological macromolecules).
\newblock {\em Ultramicroscopy}, 13(1–2):165–183.

\bibitem[van Heel, 1987]{VanHeel1987Angular}
van Heel, M. (1987).
\newblock Angular reconstitution: a posteriori assignment of projection
  directions for {3D} reconstruction.
\newblock {\em Ultramicroscopy}, 21(2):111--123.

\bibitem[Van~Heel, 1987]{VanHeel1987Similarity}
Van~Heel, M. (1987).
\newblock Similarity measures between images.
\newblock {\em Ultramicroscopy}, 21(1):95--100.

\bibitem[Vaswani et~al., 2017]{Vaswani2017Attention}
Vaswani, A., Shazeer, N., Parmar, N., Uszkoreit, J., Jones, L., Gomez, A.~N.,
  Kaiser, L., and Polosukhin, I. (2017).
\newblock Attention is all you need.
\newblock In {\em Proc. NeurIPS}, volume~30.

\bibitem[Vulovi\'{c} et~al., 2013]{Vulovic2013ImageFormation}
Vulovi\'{c}, M., Ravelli, R.~B., van Vliet, L.~J., Koster, A.~J., Lazi\'{c},
  I., L\"{u}cken, U., Rullg\r{a}rd, H., \"{O}ktem, O., and Rieger, B. (2013).
\newblock Image formation modeling in cryo-electron microscopy.
\newblock {\em Journal of Structural Biology}, 183(1):19--32.

\bibitem[Wright et~al., 2025]{Aspire2025}
Wright, G., Andén, J., Bansal, V., Xia, J., Langfield, C., Carmichael, J.,
  Sowattanangkul, K., Brook, R., Shi, Y., Heimowitz, A., Pragier, G., Sason,
  I., Moscovich, A., Shkolnisky, Y., and Singer, A. (2025).
\newblock Computationalcryoem/aspire-python: v0.13.2.

\bibitem[Wu and He, 2018]{Wu2018GroupNorm}
Wu, Y. and He, K. (2018).
\newblock Group normalization.
\newblock In {\em Proc. ECCV}, pages 3--19.

\bibitem[wwPDB consortium, 2018]{PDB2018}
wwPDB consortium (2018).
\newblock Protein data bank: the single global archive for {3D} macromolecular
  structure data.
\newblock {\em Nucleic Acids Research}, 47(D1):D520–D528.

\bibitem[Zhang et~al., 2017]{Zhang2017DnCNN}
Zhang, K., Zuo, W., Chen, Y., Meng, D., and Zhang, L. (2017).
\newblock Beyond a {Gaussian} denoiser: Residual learning of deep {CNN} for
  image denoising.
\newblock {\em IEEE Transactions on Image Processing}, 26(7):3142--3155.

\bibitem[Zhao and Singer, 2014a]{Zhao2014ClassAvg}
Zhao, Z. and Singer, A. (2014a).
\newblock Rotationally invariant image representation for viewing direction
  classification in cryo-{EM}.
\newblock {\em Journal of Structural Bioloy}, 186(1):153--166.

\bibitem[Zhao and Singer, 2014b]{Zhao2014Rotationally}
Zhao, Z. and Singer, A. (2014b).
\newblock Rotationally invariant image representation for viewing direction
  classification in {cryo-EM}.
\newblock {\em Journal of Structural Biology}, 186(1):153--166.

\bibitem[Zhong et~al., 2021]{Zhong2021CryoDRGN}
Zhong, E.~D., Bepler, T., Berger, B., and Davis, J.~H. (2021).
\newblock {CryoDRGN}: reconstruction of heterogeneous {cryo-EM} structures
  using neural networks.
\newblock {\em Nature Methods}, 18(2):176--185.

\end{thebibliography}

\clearpage

\appendix
\section{Inverting the polar representation}
\label{sec:convoApprox}

For the polar grid, the radii $r_0, r_1, \ldots, r_{N-1}$ are given by a Gauss--Jacobi quadrature rule over $[0, \sqrt{2} + \Delta]$ with parameters $(0, 1)$~\citep[Chapter 4.8--1]{Ralston2001NumericalAnalysis} for some $\Delta \ge 0$.
We denote the corresponding quadrature weights by $w_0, w_1, \ldots, w_{N-1}$.
The angles are given by $\gamma_m = 2\pi m / M$ for $m = 0, 1, \ldots, {M-1}$.
The resulting grid $\{(u_{nm}, v_{nm})\}_{nm}$ is then given by~(see Figure~\ref{fig:polar-decomp-grid})
\begin{equation}
    \label{eq:polar-grid}
        u_{nm} = r_n \cos \gamma_m \quad v_{nm} = r_n \sin \gamma_m.
\end{equation}

Applying the adjoint of~\eqref{eq:polar-decomp} to some polar image $z[n, m]$ gives back a Cartesian image $P^\transp z[i, j]$.
This does not recover the original image $x$, but the following proposition shows that it can be approximated as a convolution for $N$ and $M$ large enough.
\begin{prop}
    \label{prop:convoApprox}
    Let
    \begin{equation}
        \phi[i, j] = \frac{1}{\pi b^2 L^2} \exp\left(-\frac{i^2 + j^2}{b^2 L^2}\right).
    \end{equation}
    Then
    \begin{equation}
        \label{eq:deconv}
        P^\transp Px[i,j] = x \star \phi[i,j] + \varepsilon,
    \end{equation}
    where $\star$ denotes a discrete 2D convolution and
    \begin{equation}
        \varepsilon = O\left(\frac{(\sqrt{2} + \Delta)^2}{N^2 b^4} + \left(\frac{M!}{2M!}\right)^2 \left(\frac{2}{b^2}\right)^{M} + b e^{-\Delta^2 / b^2} \right).
    \end{equation}
\end{prop}

\begin{proof}
To see why $P^\transp P$ can be approximated with a discrete convolution, let us write out the full expression
\begin{align}
    \nonumber
    &P^\transp Px[p, q] \\
    \nonumber
    &\quad= Z^{-2} \sum_{n=0}^{N-1} w_n \sum_{m=0}^{M-1} e^{-\frac{(u_{nm} - 2p)^2 + (v_{nm} - 2q)^2}{2b^2}} \\
    &\quad\quad\sum_{i,j=-L/2}^{L/2-1} x[i, j] e^{-\frac{(u_{nm}-2i/L)^2 + (v_{nm} - 2j/L)^2}{2b^2}} \\
    \nonumber
    &\quad= Z^{-2} \sum_{i,j=-L/2}^{L/2-1} 2\pi M e^{-\frac{(2i/L - 2p/L)^2 + (2j/L - 2q/L)^2}{4b^2}}x[i, j] \\
    &\quad\quad\sum_{n=0}^{N-1} \sum_{m=0}^{M-1} \frac{w_n}{2\pi M} e^{- \frac{((i+p)/L - u_{nm})^2 + ((j+q)/L - v_{nm})^2}{b^2}},
\end{align}
where we have used the identity
\begin{equation}
    (a - c)^2 + (b - c)^2 = \frac{1}{2} (a - b)^2 + 2\left(\frac{a+b}{2} - c\right)^2.
\end{equation}

We now set $x = (i+p) / L$ and $y = (j+q)/L$ and note that $(x, y) \in [-1, 1]^2$.
Recall that we have $u_{nm} = r_n \cos(\alpha_m)$ and $v_{nm} = r_n \sin(\alpha_m)$ as described in \eqref{eq:polar-grid}.
The sum over $n$ and $m$ above can therefore be written as
\begin{equation}
    \sum_{n=0}^{N-1} \sum_{m=0}^{M-1} \frac{w_n}{2\pi M} e^{-\frac{(x-r_n \cos(\alpha_m))^2 + (y - r_n \sin(\alpha_m))^2}{b^2}}.
\end{equation}

We first consider the sum over the angular index $m$ for a fixed $n$.
This is a sum of the periodic function $s_n(\alpha) = e^{-\frac{(x-r_n \cos(\alpha)^2 + (y-r_n \sin(\alpha))^2}{b^2}}$ sampled on a uniform grid over $[0, 2\pi)$ of size $M$.
Using a Fourier series decomposition of $s_n(\alpha)$, we can see that this sum is equal to
\begin{equation}
    \int_{0}^{2\pi} s_n(\alpha) d\alpha + \epsilon,
\end{equation}
where $|\epsilon| \le \frac{C}{N^2} \int_{0}^{2\pi} |s_n''(\alpha)| d\alpha$ for some constant $C$.
Differentiating $s_n(\alpha)$ and using the fact that $r_n < \sqrt{2} + \Delta$, we have that $|s_n''(\alpha)| \le C(\sqrt{2} + \Delta)r_n$ for some constant $C$.
This then gives
\begin{align}
    \nonumber
    &\sum_{m=0}^{M-1} \frac{1}{2\pi M} e^{-\frac{(x-r_n \cos(\alpha_m))^2 + (y-r_n\sin(\alpha_m))^2}{b^2}} \\
    &\quad= \sum_{m=0}^{M-1} \frac{1}{2\pi M} s(\alpha_m) \\
    &\quad= \int_{0}^{2\pi} s_n(\alpha) d\alpha + \epsilon(r_n)
\end{align}
where $|\epsilon(r)| \le C(\sqrt{2} + \Delta)r/b^4 N^2$.

We now multiply by $w_n$ and sum over $n$.
Since the error term $\epsilon(r)$ can be bounded by a linear function in $r$ and the Gauss--Jacobi quadrature integrates polynomials of degree $2M-1$ exactly, we have
\begin{align}
    \left|\sum_{m=0}^{M-1} w_n \epsilon(r_n) \right| &\le \frac{C(\sqrt{2} + \Delta)}{b^4 N^2} \int_{0}^{\sqrt{2} + \Delta} r^2 dr \\ &= \frac{C(\sqrt{2} + \Delta)^4}{b^4 N^2}.
\end{align}

What remains is thus to calculate
\begin{equation}
    \sum_{n=0}^{N-1} w_n I(r_n),
\end{equation}
where $I(r) = \int_0^{2\pi} s_n(\alpha) d\alpha$.
For the Gauss--Jacobi quadrature rule used here, the error is proportional to
\begin{equation}
    \frac{M!}{2M!} \left(\frac{(M+1)!}{(2M+1)!}\right)^2 I^{(2M)}(r)
\end{equation}
for some $r \in [0, \sqrt{2} + \Delta]$.
Our goal is therefore to bound the $2M$th derivative of $I(r)$.

We first note that $(x, y)$ can be written in a polar representation, obtaining $x = \rho \cos(\eta)$ and $y = \rho \sin(\eta)$ for $\rho \in [0, \sqrt{2}]$ and $\eta \in [0, 2\pi)$.
This allows us to rewrite $I(r)$ in the following manner
\begin{align}
    I(r) &= \int_0^{2\pi} e^{-\frac{(x-r\cos(\alpha))^2 + (y-r\sin(\alpha))^2}{b^2}} d\alpha \\
    &= \int_0^{2\pi} e^{-\frac{(r - \rho\cos(\alpha - \eta))^2}{b^2}} \cdot e^{-\frac{\rho^2 \sin^2(\alpha - \eta)}{b^2}} d\alpha \\
    &= \int_0^{2\pi} e^{-\frac{(r - \rho\cos(\alpha))^2}{b^2}} \cdot e^{-\frac{\rho\sin^2(\alpha)}{b^2}} d\alpha.
\end{align}
Setting $t(r) = e^{-\frac{(r-\rho\cos(\alpha)^2}{b^2}}$, we note that this is simply an affine transformation of the Gaussian function $r \mapsto e^{-r^2}$.
Consequently, its derivatives can be expressed using Hermite polynomials.
Specifically, if $H_k$ is the $k$th Hermite polynomial, we have that
\begin{equation}
    t^{(k)}(r) = \frac{(-1)^k}{b^k} H_k\left(\frac{r - \rho\cos(\alpha)}{b}\right) e^{-\frac{(r-\rho\cos(\alpha))^2}{b^2}}.
\end{equation}
From standard bounds on $H_k$ \citep[{(18.14.9)}]{DLMF}, we obtain that
\begin{equation}
    |t^{(k)}(r)| \le \frac{\sqrt{2^k k!}}{b^k}.
\end{equation}
Computing the $k$th derivative of $I(r)$, plugging in the above bound, and noting that the second factor in the integrand (which does not depend on $r$) is less than one, we obtain
\begin{equation}
    |I^{(k)}(r)| \le \frac{\sqrt{2^k k!}}{b^k}.
\end{equation}
The Gaussi--Jacobi quadrature error can therefore be bounded by
\begin{align}
    \nonumber
    &\frac{M!}{2M!} \left(\frac{(M+1)!}{(2M+1)!}\right)^2 \frac{\sqrt{2^{2M} 2M!}}{b^{2M}} \\
    &\quad= \frac{M!}{\sqrt{2M!}} \left(\frac{(M+1)!}{(2M+1)!}\right)^2 \left(\frac{2}{b^2}\right)^M \\
    &\quad\le \left(\frac{M!}{2M!}\right)^2 \left(\frac{2}{b^2}\right)^M.
\end{align}

Combining these results, we obtain that
\begin{align}
    \nonumber
    &\sum_{n=0}^{N-1} \sum_{m=0}^{M-1} \frac{w_n}{2\pi M} e^{-\frac{(x-r_n \cos(\alpha_m))^2 + (y - r_n \sin(\alpha_m))^2}{b^2}} \\
    \nonumber
    &\quad = \int_D e^{-\frac{(x - u)^2 + (y - v)^2}{b^2}} du dv \\
    &\quad\quad + O\left(\frac{(\sqrt{2} + \Delta)^2}{N^2 b^4} + \left(\frac{M!}{2M!}\right)^2 \left(\frac{2}{b^2}\right)^{M}\right)
\end{align}
for $D = \{(u,v) \mid u^2 + v^2 < (\sqrt{2} + \Delta)^2\}$.
We now need to approximate this integral.
To do this, we extend it to an integral over all of $\mathbb{R}^2$, which gives the result $\pi b^2$.
To quantify the error in the approximation, we must therefore bound
\begin{equation}
    \int_{\mathbb{R}^2 \setminus D} e^{-\frac{(x - u)^2 + (y - v)^2}{b^2}} du dv.
\end{equation}

We again consider polar coordinates for both $(x,y)$ and $(u,v)$, which transforms the above expression into
\begin{align}
    \nonumber
    &\int_{\sqrt{2} + \Delta}^{+\infty} \int_{0}^{2\pi} e^{-\frac{r^2 + \rho^2 - 2r\rho \cos(\alpha - \eta)}{b^2}} d\alpha dr \\
    &\quad= \int_{\sqrt{2} + \Delta}^{+\infty} e^{-\frac{r^2 + \rho^2}{b^2}} \int_{0}^{2\pi} e^{\frac{2r\rho \cos(\alpha)}{b^2}} d\alpha dr.
\end{align}
The innermost integral can be written as $2\pi I_0(2r\rho/b^2)$, where $I_0$ is the zeroth-order modified Bessel function of the first kind.
This gives
\begin{equation}
    2\pi \int_{\sqrt{2} + \Delta}^{+\infty} e^{-\frac{r^2 + \rho^2}{b^2}} I_0(2r\rho/b^2) dr.
\end{equation}
Bounding $I_0(2r\rho/b^2)$ by $e^{2r\rho/b^2}$ \citep[{(10.14.3)}]{DLMF}, we obtain the upper bound
\begin{align}
    \nonumber
    &2\pi \int_{\sqrt{2} + \Delta}^{+\infty} e^{-\frac{r^2 + \rho^2}{b^2}} e^{\frac{2r\rho}{b^2}} dr \\
    &\quad= 2\pi \int_{\sqrt{2} + \Delta}^{+\infty} e^{-\frac{(r - \rho)^2}{b^2}} dr = 2\pi \int_{\sqrt{2} + \Delta - \rho}^{+\infty} e^{-\frac{r^2}{b^2}} dr.
\end{align}
Since $\rho \le \sqrt{2}$, we have that $\sqrt{2} + \Delta - \rho \ge 0$, so we can bound this integral using standard results on the integral of a Gaussian function \citet[{(7.8.3)}]{DLMF} to obtain the bound
\begin{equation}
    2\pi b e^{\frac{(\sqrt{2} + \Delta - \rho)^2}{b^2}} \le 2\pi b e^{-\frac{\Delta^2}{b^2}}.
\end{equation}

This, together with the quadrature error bounds, gives us that
\begin{align}
    \nonumber
    &\sum_{n=0}^{N-1} \sum_{m=0}^{M-1} \frac{w_n}{2\pi M} e^{-\frac{(x-r_n \cos(\alpha_m))^2 + (y - r_n \sin(\alpha_m))^2}{b^2}} \\
    \nonumber
    &\quad = \pi b^2 \\
    &\quad\quad + O\left(\frac{(\sqrt{2} + \Delta)^2}{N^2 b^4} + \left(\frac{M!}{2M!}\right)^2 \left(\frac{2}{b^2}\right)^{M} + b e^{-\frac{\Delta^2}{b^2}} \right)
\end{align}
Plugging this into our expression for $P^\transp Px$ then gives us the desired result.
\end{proof}

Because of the above result, we can reconstruct $x$ by solving the deconvolution problem in~\eqref{eq:deconv}.
This can be done in several ways, most easily by approximating the discrete convolution with a circular convolution and solving it by pointwise division in the Fourier domain~\citep{Bertero2021InverseProblems}.
To ensure that this deconvolution problem is relatively well-posed, however, the Fourier transform of $\phi$ must not decay too fast -- in other words, we cannot choose $b$ to be too small.
We have found that choosing $b$ on the order of $1/L$ results in a well-conditioned deconvolution problem.

To evaluate the accuracy of this process, we computed the relative MSE obtained when converting to and from the polar representation.
This was done on the validation dataset used in our experiments containing clean 100 000 images, each of size $L = 64$.
We obtain the following results:
\begin{center}
    \begin{tabular}{llll}
        $N$ & $M$ & $b$ & Relative MSE \\
        \hline
        $L$ & $4L$ & $1/L$ & $0.00019$ \\  
        $L$ & $4L$ & $1/2L$ & $0.06333$ \\  
        $L$ & $4L$ & $2/L$ & $0.66780$ \\  

        $L$ & $3L$ & $1/L$ & $0.00019$ \\  
        $L$ & $2L$ & $1/L$ & $0.00022$ \\  
        $L$ & $3L/2$ & $1/L$ & $0.00117$ \\  

        $L/2$ & $4L$ & $1/L$ & $0.77263$ \\  
        $3L/4$ & $4L$ & $1/L$ & $0.00194$ \\  
        $3L/4$ & $3L/2$ & $1/L$ & $0.00294$     
    \end{tabular}
\end{center}

Given the low error and short running time (see Appendix~\ref{sec:compComplexity}), we will choose the last configuration ($N = 3L/4$, $M = 3L/2$, $b = 1/L$) for the remainder of the experiments in this work.

\section{Properties of $\phi$}
\label{sec:sumPhi}

Proposition~\ref{prop:convoApprox} shows that $P$ is an approximate isometry when restricted to smooth images.
Indeed, we can show that $\phi[i,j] = (\pi b^2 L^2)^{-1} e^{-\frac{i^2 + j^2}{b^2 L^2}}$ is a lowpass filter which sums to one.

For this, we use the Poisson summation formula~\citep[Chapter VII.2]{Stein1971Fourier}.
First, we note that the Fourier transform of the continuous function
\begin{equation}
    \phi(u, v) = \frac{1}{\pi b^2 L^2} e^{-\frac{u^2 + v^2}{b^2 L^2}}
\end{equation}
is given by
\begin{equation}
    \widehat{\phi} (\omega, \xi) = e^{-\pi^2 b^2 L^2 (\omega^2 + \xi^2)}.
\end{equation}
We thus have
\begin{align}
    \nonumber
    \sum_{i,j=-\infty}^{+\infty} \phi[i, j] &= \sum_{i,j=-\infty}^{+\infty} \phi(i, j) \\
    &= \sum_{k,\ell=-\infty}^{+\infty} \widehat{\phi} (k, \ell) \\
    &= 1 + \sum_{k,\ell \neq 0} e^{-\pi^2 b^2 L^2 (k^2 + \ell^2)}.
\end{align}
Provided that $b$ is large enough, the infinite sum is negligible (for $b = 1/L$, the largest term is of the order $10^{-4}$).
We thus have the desired result.

This in turn then means that
\begin{equation}
    (P y)^\transp (P x) = y^\transp P^\transp P x \approx y^\transp (x \star \phi) \approx y^\transp x.
\end{equation}
for two smooth images $x$ and $y$.

To prove stability, we note that since the Fourier transform of $\phi$ has its largest value close to one, the largest eigenvalue of $P^\transp P$ is also close to one, which means that $\|P\| \approx 1$.

We now evaluate this numerically in two experiments.
First, we take the clean 100 000 images in the validation dataset and decompose them in our proposed RBF-based polar representation.
We then do the same using an interpolation-based polar representation implemented using \texttt{scikit.transform.warp}, which has options for zeroth-, first-, second-, and third-order interpolation, corresponding to nearest-neigbor, linear, quadratic, and cubic interpolation, respectively.
Each original (Cartesian) image is then transformed by a small perturbation (rotation by ten degrees or shift by one pixel) and its polar representation is computed.
Finally, we compute the factor $\|Px - Py\| / \|x - y\|$, where $x$ is the original image, $y$ is the transformed image, and $P$ is the chosen polar mapping.
This should give us an indication of the norm of $\|P\|$ (indeed, each provides a lower bound) for these various mappings.

The results are as follows:
\begin{center}
    \begin{tabular}{l|ll}
        Transformation & Rotation & Shift \\
        \hline
        Polar RBF & 0.84 & 0.84 \\
        \texttt{warp}, nearest neighbor & 3.40 & 4.01 \\
        \texttt{warp}, linear & 2.97 & 3.48 \\
        \texttt{warp}, quadratic & 3.29 & 3.85 \\
        \texttt{warp}, cubic & 3.31 & 3.88
    \end{tabular}
\end{center}
We see, therefore, that not only does $\|P\| \approx 1 $ seem reasonable, but the alternative methods based on standard interpolation yield much higher norms, signifying lower robustness to perturbations.
This in turn results in worse training behavior when incorporated into a neural network structure.

\section{Computational complexity of polar mapping}
\label{sec:compComplexity}

In terms of computational cost, application of $P$ and $P^\transp$ can be implemented efficiently using fast Gaussian gridding~\citep{Greengard2004NUFFT, Barnett2019FINUFFT, Shih2021cuFINUFFT}.
As a result of this, the number of non-negligible terms in \eqref{eq:polar-decomp} and its adjoint is $\bigO(1)$, which means that both operations can be computed in $\bigO(L^2)$ time (recall that $N = \bigO(L)$ and $M = \bigO(L)$).
Finally, the deconvolution step, if implemented using FFTs, has a computational cost of $\bigO(L^2 \log L)$.
Mapping between Cartesian and polar domains can thus be achieved quite efficiently.

We now repeat the experiment of Appendix~\ref{sec:convoApprox} and time the round-trip pipeline (decomposing and reconstructing) for the 100 000 images in the validation set.
This gives the following results on an A100 GPU:
\begin{center}
    \begin{tabular}{llll}
        $N$ & $M$ & $b$ & Time per image \\
        \hline
        $L$ & $4L$ & $1/L$ & 40 \textmu{}s \\
        $L$ & $4L$ & $1/2L$ & 16 \textmu{}s \\
        $L$ & $4L$ & $2/L$ &  128 \textmu{}s \\

        $L$ & $3L$ & $1/L$ & 31 \textmu{}s \\
        $L$ & $2L$ & $1/L$ & 22 \textmu{}s \\
        $L$ & $3L/2$ & $1/L$ & 17 \textmu{}s \\

        $L/2$ & $4L$ & $1/L$ & 22 \textmu{}s \\
        $3L/4$ & $4L$ & $1/L$ & 31 \textmu{}s \\
        $3L/4$ & $3L/2$ & $1/L$ & 14 \textmu{}s
    \end{tabular}
\end{center}
Again, we see that the last configuration gives quite a good running time while also achieving low error (see Appendix~\ref{sec:convoApprox}).

\section{Schematics of networks}
\label{sec:architectureSchematics}

The architectures of the polar CNN and polar transformer are illustrated in Figures~\ref{fig:polarCNN} and \ref{fig:polarTransformer}, respectively.

\begin{figure*}
    \begin{center}
        \begin{tikzpicture}
            \newcommand{\myheight}{10mm}

            \node (cartesian-input) at (0, 2) {\includegraphics[height=\myheight]{sample_noisy_image}};

            \node (cartesian-input-label) [below=5mm of cartesian-input] {$L\times L$};

            \node (polar-input) at (2.5, 2) {\includegraphics[height=\myheight]{sample_noisy_polar_image}};
            \node (polar-input-label) [below=5mm of polar-input] {$N \times M$};

            \node (polar-layer-one-a) at (5.5, 1.8) {\includegraphics[height=\myheight]{sample_noisy_polar_image}};
            \node (polar-layer-one-b) at (5.6, 1.9) {\includegraphics[height=\myheight]{sample_noisy_polar_image}};
            \node (polar-layer-one-c) at (5.7, 2.0) {\includegraphics[height=\myheight]{sample_noisy_polar_image}};
            \node (polar-layer-one-d) at (5.8, 2.1) {\includegraphics[height=\myheight]{sample_noisy_polar_image}};
            \node (polar-layer-one-e) at (5.9, 2.2) {\includegraphics[height=\myheight]{sample_noisy_polar_image}};

            \node (polar-layer-one-c-label) [below=5mm of polar-layer-one-c] {$C \times N \times M$};

            \node at (7.5, 2) {$\cdots$};

            \node (polar-layer-two-a) at (9, 1.8) {\includegraphics[height=\myheight]{sample_polar_image}};
            \node (polar-layer-two-b) at (9.1, 1.9) {\includegraphics[height=\myheight]{sample_polar_image}};
            \node (polar-layer-two-c) at (9.2, 2.0) {\includegraphics[height=\myheight]{sample_polar_image}};
            \node (polar-layer-two-d) at (9.3, 2.1) {\includegraphics[height=\myheight]{sample_polar_image}};
            \node (polar-layer-two-e) at (9.4, 2.2) {\includegraphics[height=\myheight]{sample_polar_image}};
            \node (polar-layer-two-label) [below=5mm of polar-layer-two-c] {$C \times N \times M$};

            \node (polar-output) at (12.5, 2) {\includegraphics[height=\myheight]{sample_polar_image}};
            \node (polar-output-label) [below=5mm of polar-output] {$N \times M$};

            \node (cartesian-output) at (15, 2) {\includegraphics[height=\myheight]{sample_image}};
            \node (cartesian-output-label) [below=5mm of cartesian-output] {$L \times L$};

            \draw[->] (cartesian-input) -- (polar-input) node [midway, above] {$P$};
            \draw[->,shorten >=2mm] (polar-input) -- (polar-layer-one-c) node [midway, above, xshift=-1mm] {$\ostar\,h_1$};
            \draw[->,shorten <=2mm] (polar-layer-two-c) -- (polar-output) node [midway, above, xshift=1mm] {$\ostar h_n$};
            \draw[->] (polar-output) -- (cartesian-output) node [midway, above] {$P^{-1}$};
        \end{tikzpicture}
    \end{center}
    \caption{
        \label{fig:polarCNN}
    Outline of the polar CNN architecture.
    The Cartesian image is first transformed into the polar representation using the mapping $P$.
    It then passes through a series of 1D convolutional layers defined by the angular filters $h_1, h_2, \ldots, h_n$, each of which inclues a ReLU activation and a group normalization step (not shown above).
    Finally, the result is converted back to a Cartesian image using $P^{-1}$.}
\end{figure*}

\begin{figure*}
    \begin{center}
        \begin{tikzpicture}
            [
                every text node part/.style={align=center},
            ]
            \newcommand{\myheight}{10mm}

            \node (input-0) [inner sep=0mm] at (0, 0) {\includegraphics[height=\myheight]{alignment/noisy_0000}};
            \node (input-1) [below=5mm of input-0, inner sep=0mm] {\includegraphics[height=\myheight]{alignment/noisy_0001}};
            \node (input-2) [below=5mm of input-1, inner sep=0mm] {\includegraphics[height=\myheight]{alignment/noisy_0002}};
            \node (input-cont) [below=5mm of input-2] {$\vdots$};
            \node (input-last) [below=5mm of input-cont, inner sep=0mm] {\includegraphics[height=\myheight]{alignment/noisy_0004}};

            \node (preproc-0-a)    [right=20mm of input-0, xshift=-1mm, yshift=-1mm] {\includegraphics[height=\myheight]{sample_noisy_polar_image}};
            \node (preproc-0-b)    [right=20mm of input-0, xshift=0mm, yshift=0mm] {\includegraphics[height=\myheight]{sample_noisy_polar_image}};
            \node (preproc-0-c)    [right=20mm of input-0, xshift=1mm, yshift=1mm] {\includegraphics[height=\myheight]{sample_noisy_polar_image}};

            \node (preproc-1-a)    [right=20mm of input-1, xshift=-1mm, yshift=-1mm] {\includegraphics[height=\myheight]{sample_noisy_polar_image}};
            \node (preproc-1-b)    [right=20mm of input-1, xshift=0mm, yshift=0mm] {\includegraphics[height=\myheight]{sample_noisy_polar_image}};
            \node (preproc-1-c)    [right=20mm of input-1, xshift=1mm, yshift=1mm] {\includegraphics[height=\myheight]{sample_noisy_polar_image}};

            \node (preproc-2-a)    [right=20mm of input-2, xshift=-1mm, yshift=-1mm] {\includegraphics[height=\myheight]{sample_noisy_polar_image}};
            \node (preproc-2-b)    [right=20mm of input-2, xshift=0mm, yshift=0mm] {\includegraphics[height=\myheight]{sample_noisy_polar_image}};
            \node (preproc-2-c)    [right=20mm of input-2, xshift=1mm, yshift=1mm] {\includegraphics[height=\myheight]{sample_noisy_polar_image}};

            \node (preproc-last-a) [right=20mm of input-last, xshift=-1mm, yshift=-1mm] {\includegraphics[height=\myheight]{sample_noisy_polar_image}};
            \node (preproc-last-b) [right=20mm of input-last, xshift=0mm, yshift=0mm] {\includegraphics[height=\myheight]{sample_noisy_polar_image}};
            \node (preproc-last-c) [right=20mm of input-last, xshift=1mm, yshift=1mm] {\includegraphics[height=\myheight]{sample_noisy_polar_image}};

            \node (attention) [draw, right=6mm of preproc-2-b, minimum height=80mm, minimum width=30mm] {Angular attention};

            \node (postproc-0-a)    [right=40mm of preproc-0-b, xshift=-1mm, yshift=-1mm] {\includegraphics[height=\myheight]{sample_polar_image}};
            \node (postproc-0-b)    [right=40mm of preproc-0-b, xshift=0mm, yshift=0mm] {\includegraphics[height=\myheight]{sample_polar_image}};
            \node (postproc-0-c)    [right=40mm of preproc-0-b, xshift=1mm, yshift=1mm] {\includegraphics[height=\myheight]{sample_polar_image}};

            \node (postproc-1-a)    [right=40mm of preproc-1-b, xshift=-1mm, yshift=-1mm] {\includegraphics[height=\myheight]{sample_polar_image}};
            \node (postproc-1-b)    [right=40mm of preproc-1-b, xshift=0mm, yshift=0mm] {\includegraphics[height=\myheight]{sample_polar_image}};
            \node (postproc-1-c)    [right=40mm of preproc-1-b, xshift=1mm, yshift=1mm] {\includegraphics[height=\myheight]{sample_polar_image}};

            \node (postproc-2-a)    [right=40mm of preproc-2-b, xshift=-1mm, yshift=-1mm] {\includegraphics[height=\myheight]{sample_polar_image}};
            \node (postproc-2-b)    [right=40mm of preproc-2-b, xshift=0mm, yshift=0mm] {\includegraphics[height=\myheight]{sample_polar_image}};
            \node (postproc-2-c)    [right=40mm of preproc-2-b, xshift=1mm, yshift=1mm] {\includegraphics[height=\myheight]{sample_polar_image}};

            \node (postproc-last-a) [right=40mm of preproc-last-b, xshift=-1mm, yshift=-1mm] {\includegraphics[height=\myheight]{sample_polar_image}};
            \node (postproc-last-b) [right=40mm of preproc-last-b, xshift=0mm, yshift=0mm] {\includegraphics[height=\myheight]{sample_polar_image}};
            \node (postproc-last-c) [right=40mm of preproc-last-b, xshift=1mm, yshift=1mm] {\includegraphics[height=\myheight]{sample_polar_image}};

            \node (output-0) [right=20mm of postproc-0-b, inner sep=0mm] {\includegraphics[height=\myheight]{alignment/clean_0000}};
            \node (output-1) [below=5mm of output-0, inner sep=0mm] {\includegraphics[height=\myheight]{alignment/clean_0001}};
            \node (output-2) [below=5mm of output-1, inner sep=0mm] {\includegraphics[height=\myheight]{alignment/clean_0002}};
            \node (output-cont) [below=5mm of output-2] {$\vdots$};
            \node (output-last) [below=5mm of output-cont, inner sep=0mm] {\includegraphics[height=\myheight]{alignment/clean_0004}};

            \draw [->,shorten <=1mm, shorten >=1mm] (input-0) -- (preproc-0-b) node [midway, above] {Pol. CNN};
            \draw [->,shorten <=1mm, shorten >=1mm] (input-1) -- (preproc-1-b) node [midway, above] {Pol. CNN};
            \draw [->,shorten <=1mm, shorten >=1mm] (input-2) -- (preproc-2-b) node [midway, above] {Pol. CNN};
            \draw [->,shorten <=1mm, shorten >=1mm] (input-last) -- (preproc-last-b) node [midway, above] {Pol. CNN};

            \draw [->,shorten <=1mm] (preproc-0-b.east) -- ++(4mm, 0);
            \draw [->,shorten <=1mm] (preproc-1-b.east) -- ++(4mm, 0);
            \draw [->,shorten <=1mm] (preproc-2-b.east) -- ++(4mm, 0);
            \draw [->,shorten <=1mm] (preproc-last-b.east) -- ++(4mm, 0);

            \draw [<-,shorten >=1mm] (postproc-0-b.west) -- ++(-4mm, 0);
            \draw [<-,shorten >=1mm] (postproc-1-b.west) -- ++(-4mm, 0);
            \draw [<-,shorten >=1mm] (postproc-2-b.west) -- ++(-4mm, 0);
            \draw [<-,shorten >=1mm] (postproc-last-b.west) -- ++(-4mm, 0);

            \draw [->,shorten <=1mm, shorten >=1mm] (postproc-0-b.east) -- (output-0.west) node [midway, above] {Pol. CNN};
            \draw [->,shorten <=1mm, shorten >=1mm] (postproc-1-b.east) -- (output-1.west) node [midway, above] {Pol. CNN};
            \draw [->,shorten <=1mm, shorten >=1mm] (postproc-2-b.east) -- (output-2.west) node [midway, above] {Pol. CNN};
            \draw [->,shorten <=1mm, shorten >=1mm] (postproc-last-b.east) -- (output-last.west) node [midway, above] {Pol. CNN};
        \end{tikzpicture}
    \end{center}
    \caption{\label{fig:polarTransformer} Outline of the polar transformer architecture.
    The Cartesian images are first preprocessed by a polar CNN (after being converted to the polar representation), then fed into the angular attention mechanism.
    This then combines the information from the various images by computing the key, and query vectors for each image (using another set of polar CNNs), calculating the attention coefficients, and using these together with the value vectors to compute the output.
    The output is then processed by another polar CNN and finally converted to a Cartesian image.
    Note that the angular attention block can be repeated, but we have found that one such block suffices for the purposes of denoising.}
\end{figure*}

\section{Angular attention}
\label{sec:angularAttention}

We can examine the behavior of the angular attention mechanism by studying the attention coefficients for a pair of projection images.
One image is the rotated version of another and both are subjected to varying levels of noise.
By running these through the polar transformer, we can record the attention coefficients and plot them according to the angular index $\ell$.
The results are shown in Figure~\ref{fig:angularAttention}.

At high SNR, the attention coefficients are concentrated around the true rotation angle, showing that the angular attention mechanism successfully aligns the two images with respct to one another.
As the noise increases, however, the attention coefficients are dispersed over a wider range of angles, reflecting the increased uncertainty about the angular registration.
Importantly, however, this is a robust degradation in the sense that the attention coefficients of high magnitude are still located around the true angle.

\begin{figure*}
\begin{subfigure}{0.40\textwidth}
\begin{center}
\subimport{figures/angularAttention}{noisy_snr0_25_ss8_angular_attention_gplt}
\caption{$\text{SNR}= 0.25$}
\end{center}
\end{subfigure}%
\begin{subfigure}{0.30\textwidth}
\subimport{figures/angularAttention}{noisy_snr0_0625_ss8_angular_attention_gplt}
\caption{$\text{SNR} = 0.06$}
\end{subfigure}%
\begin{subfigure}{0.30\textwidth}
\subimport{figures/angularAttention}{noisy_snr0_015625_ss8_angular_attention_gplt}
\caption{$\text{SNR}= 0.02$}
\end{subfigure}
    \caption{\label{fig:angularAttention} The angular attention coefficients for a pair of images at different noise levels are shown. For high SNR (a), we see high concentration around the true angle, while at lower SNR (b,c), there is more uncertainty regarding the correct angular assignment.}
\end{figure*}

\section{Attention for Clustering}
\label{sec:clusteringAttention}
The polar transformer architecture is designed to combine information from precisely those image in an input set which match up to plane rotation.
It is the process of finding out which ones these are that we call clustering.

How it occurs can be observed by studying the activations in the activation coefficient matrix as a set with known split of two directions is processed.
Without loss of generality (because the attention mechanism is equivariant under permutations of the images), this can be realized by taking the first 8 projections in one directions and the remaining 8 projections in another direction.
Ideally, the attention (summed over the angular direction) should then have block-matrix structure, as only the images corresponding to the same direction interact with each other.
Indeed this matches the observations in the case of relatively high SNR (see Figure~\ref{fig:attentionMatrixExamples}).
Consequently, the model can then also combine information for denoising purposes just as well as if the directions had already been known a priori, as is the case of a directional set.
This manifests in MSE scores that are almost as good with two directions as with a single direction (see Table~\ref{table:results}).

At lower SNR meanwhile, the attention mechanism cannot be as confident in classifying the directions anymore and the block structure of the matrix is more fuzzy.
This manifests in MSE values that are not quite as good anymore as in the ideal predetermined case, but there is still more information to be exploited even from the uncertain classification as from only single images and thus even the clustering model still denoises better at low SNR than single-image models can.

\section{Poisson noise parametrization}
\label{sec:poissonNoiseDetails}
For Gaussian noise, we used white Gaussian noise at a fixed noise level that was added to the images.
Poisson noise meanwhile is a quantization process, which does not directly match the notion of separate signal and noise level.
It can however be unified by introducing a \emph{Poissonicity} parameter $\eta$ which controlled the extent to which it approached a Gaussian noise (which occurred as $\eta \rightarrow 0$).

The (discrete) Poisson distribution has one parameter \(\lambda\) and a probability mass function of
\newcommand{\Pois}{\mathrm{Pois}}
\begin{equation}
  \operatorname{PMF}_{\Pois(\lambda)}(k) = \lambda^k \frac{e^{-\lambda}}{k!}
\end{equation}
The mean of this distribution is again \(\lambda\), and the standard deviation \(\sqrt{\lambda}\).

To use the Poisson distribution directly as a source of shot noise would require comparing concrete finite electron counts, which is inconvenient for comparison with the Gaussian setting as that corresponds to infinite electron rate.
Also, the Poisson expectation is always positive, whereas we use by convention noise that is symmetric about zero in the regions where the electron beam is unobstructed by molecules (the background).

We therefore define a slightly different process to emulate shot noise, which has initially two parameters \(\eta\) and \(\lambda_0\).
Given an input signal \(x\) which is close to zero in the background, let
\begin{equation}\label{eq:PoissonRescaling}
  y = \eta\cdot(\lambda_0 - c),
\end{equation}
with
\begin{equation}\label{eq:PoissonDrawing}
  c \leftarrow \Pois(\lambda_0 - x/\eta).
\end{equation}
The subtraction \(\lambda_0 - x/\eta\) expresses that positive \(x\) attenuates the pixel-wise dose, corresponding to blocking of the electrons by the molecule (more accurately, destructive interference due to phase shift).
Note that the opposite behavior can be achieved by choosing negative \(\eta\).

Then the noisy image has in the background again mean 0, matching the mean in the Gaussian case, and standard deviation
\begin{equation}\label{eq:PoissonNoiseLevel}
  N = \sqrt{\lambda_0} \cdot \eta.
\end{equation}
The local perturbation of the mean due to \(x\) has unity gain (because the multiplication by \(\eta\) in \eqref{eq:PoissonRescaling} is canceled by division in \eqref{eq:PoissonDrawing}).

We want \(N\) to match the noise's standard deviation in the Gaussian setting, which is by convention represented as
\newcommand{\SNR}{\mathrm{SNR}}
\begin{equation}\label{eq:PoissonSNRNotionConvention}
  N = \frac{S}{\sqrt\SNR}.
\end{equation}
This can be fulfilled by solving \eqref{eq:PoissonNoiseLevel} for \(\lambda_0\), which gives
\begin{equation}
  \lambda_0 = \frac{S^2}{\SNR\cdot\eta^2}.
\end{equation}
\(\eta\) remains as a parameter. If \(\eta\approx0\) is chosen, \(\lambda_0\) approaches infinity. A Poisson distribution with large \(\lambda\) is well approximated by a Gaussian distribution, in that sense the results can directly be compared.


\begin{table}
\newcommand{\exmplDenoisProj}[3]{\includegraphics[width=12mm]{figures/denoiserExamples/snr#1/#2_#3.png}}
 \newcommand{\poissonsnrexample}[2]{\exmplDenoisProj{0.03125_poissonic0.5}{#1}{#2}}
 \begin{tabular}{m{1.5cm}m{1.5cm}|m{1.5cm}m{1.5cm}}
  & & Denoiser:
   \\ (Clean) & (Poisson) & Gaussian polar tfmr. (dir.) & Poisson polar tfmr. (dir.)
   \\ \poissonsnrexample{clean}{0006}
      & \poissonsnrexample{noisy}{0006}
      & \poissonsnrexample{ss8_grpNrm4}{0006}
      & \poissonsnrexample{ss8_grpNrm4_poissonic0.5}{0006}
   \\ \poissonsnrexample{clean}{0015}
      & \poissonsnrexample{noisy}{0015}
      & \poissonsnrexample{ss8_grpNrm4}{0015}
      & \poissonsnrexample{ss8_grpNrm4_poissonic0.5}{0015}
  \\ & & \attentionGNMeanMSE{0.03125_poissonic0.5}
             & \attentionGNPoissMeanMSE{0.03125_poissonic0.5}
  \\ & & Mean MSE
 \end{tabular}
\caption{Example results like in Figure~\ref{fig:denoisedExamples}, but for Poisson noise. SNR of $0.03$ by the notion of \autoref{eq:PoissonSNRNotionConvention}.
  A transformer model that has been trained on Gaussian noise is compared with one trained on Poisson noise.}
  \label{tab:poissondenoisedExamples}
\end{table}

\begin{figure*}[t]
 \begin{center}
    \begin{tikzpicture}
     \node (att125) at (0, 2.6) {\includegraphics[width=40mm]{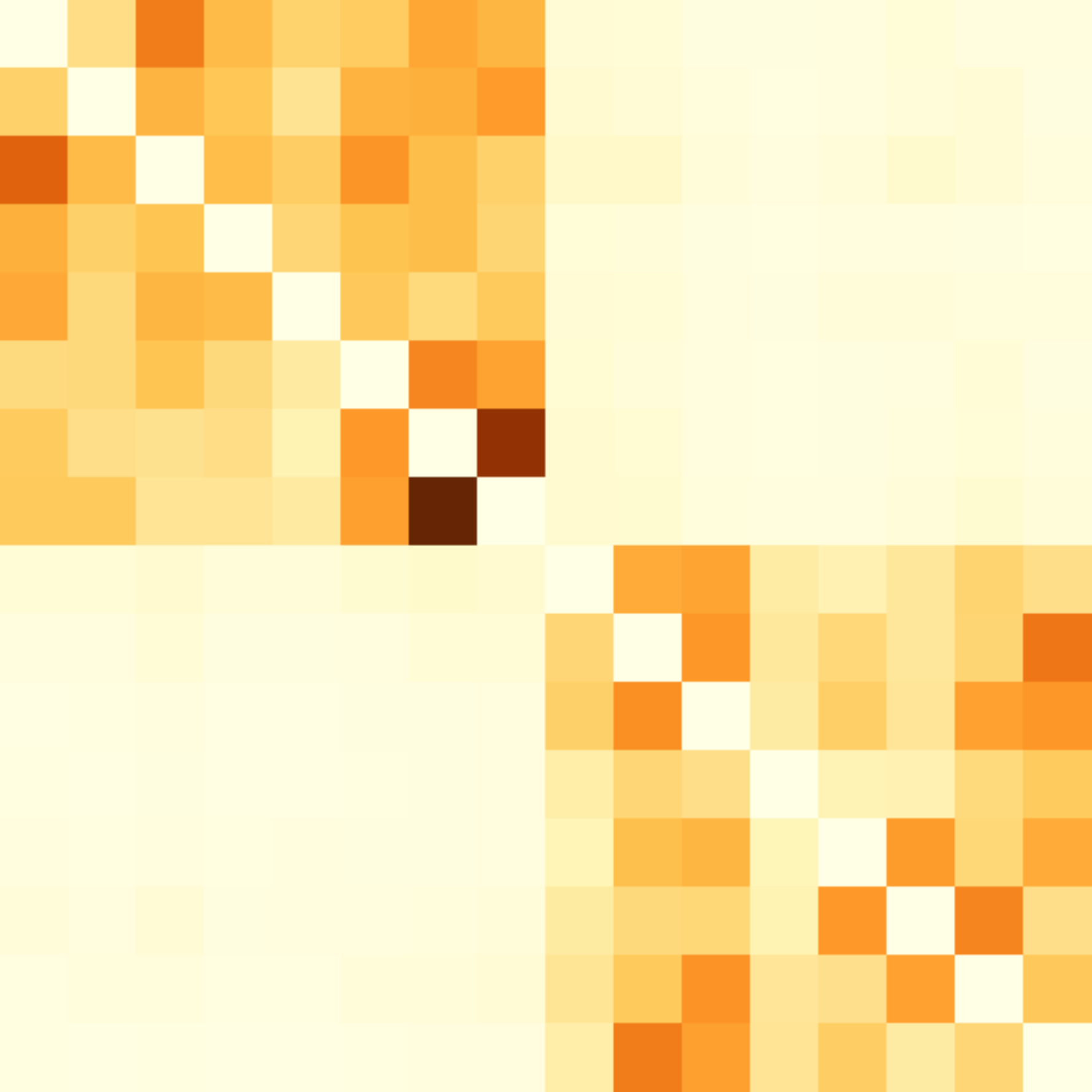}};
     \node (snr125) at (0, 0) {\(\tfrac18\)};
     \node (att03125) at (5, 2.6) {\includegraphics[width=40mm]{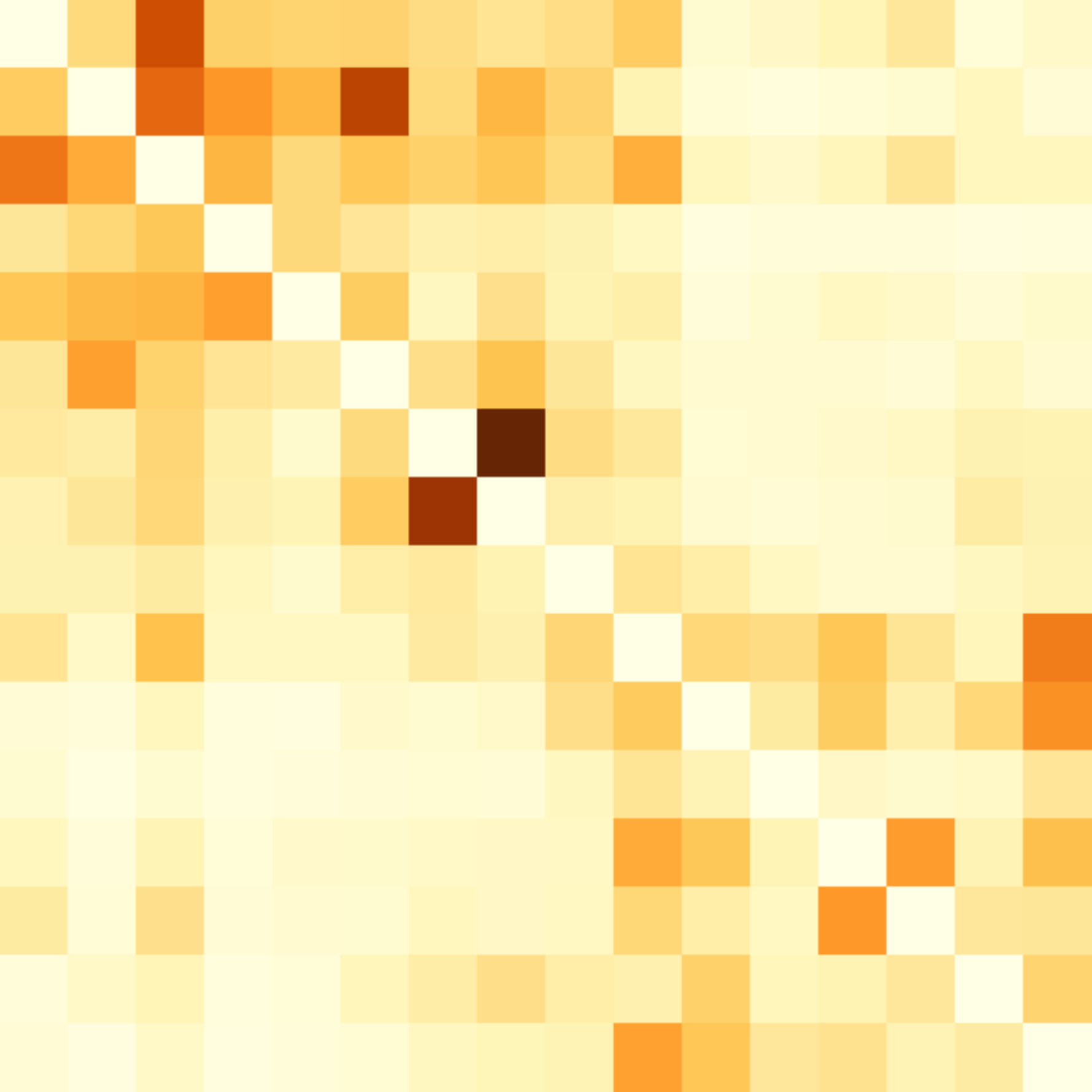}};
     \node (snr03125) at (5, 0) {\(\tfrac1{32}\)};
     \node (att0078125) at (10, 2.6) {\includegraphics[width=40mm]{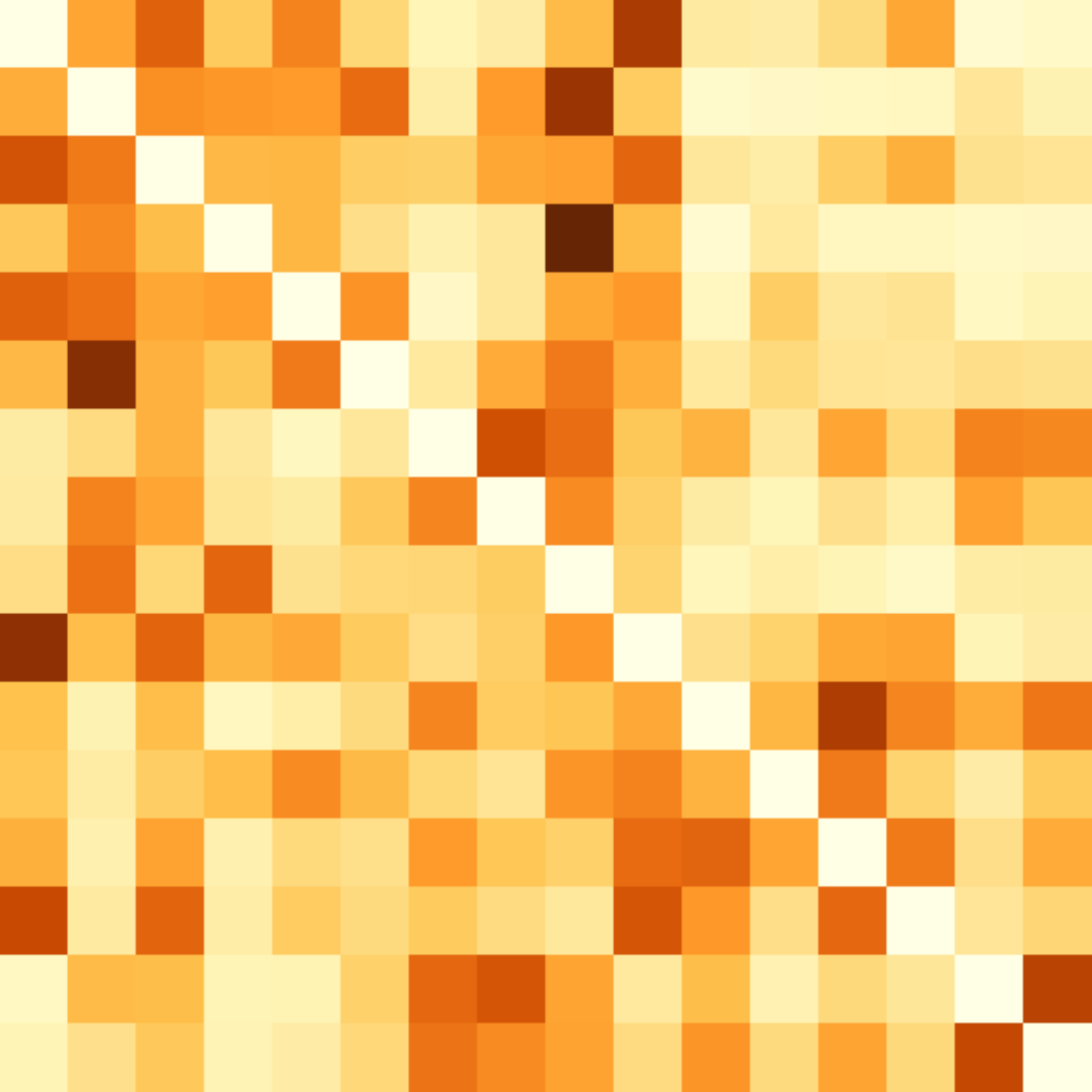}};
     \node (snr0078125) at (10, 0) {\(\tfrac1{128}\)};
    \end{tikzpicture}
 \end{center}
 \caption{\label{fig:attentionMatrixExamples} Example attention matrices for a 2-direction, 16-image-set polar transformer model at different SNR levels.}
\end{figure*}

\section{Additional results}
\label{sec:moreResults}

To further investigate the performance of the proposed polar transformer architecture, we consider some additional experiments.

\begin{table*}
    \begin{center}
        \begin{tabular}{l|cccc}
            Method & Relative MSE (\textdownarrow) & PSNR (\textuparrow) & SSIM (\textuparrow) & FRC resolution (\textdownarrow) \\
            \hline
            DnCNN & 0.088 & 29.2 & 0.95 & 3.18 \\
            U-Net & 0.089 & 29.2 & 0.95 & 3.17 \\
            Polar CNN & 0.081 & 29.6 & 0.95 & 3.19 \\
            \hline
            Polar transformer (dir.) & 0.042 & 32.2 & 0.97 & 2.68 \\
            Polar transformer (gen.) & 0.049 & 31.6 & 0.97 & 2.81 \\
        \end{tabular}
    \end{center}
    \caption{\label{table:detailedResults} More detailed denoising results for $\text{SNR} = 0.02$ with CTF and shifts. PSNR and SSIM follow the usual definitions while the FRC resolution measures the size of the smallest detail recovered in the image.}
\end{table*}

\paragraph{Other metrics}
Table~\ref{table:detailedResults} shows additional metrics for the models considered in Section~\ref{sec:results}.
In particular, we calculate the peak SNR~(PSNR) and structural similarity index measure~(SSIM) the standard formulas (for SSIM, we used the implementation of \texttt{skimage.metrics}).
These show a similar behavior as that of the relative MSE, with the transformer architectures providing a significant improvement in image quality by incorporating information from multiple images.

In addition to these standard measures, we have also calculated the Fourier ring correlation~(FRC) between the denoised images and the ground truth~\cite{VanHeel1987Similarity}.
This compares the Fourier spectrum of the two images and calculates the correlation coefficient along various rings, showing how the spectral content of the image is recovered.
These FRC plots are shown in Figure~\ref{fig:standardFrcComparison}.

To summarize the FRCs, a simple resolution measure can be calculated by taking the intersection of the FRC curve with the value $0.5$.
Inverting the frequency value to get a size measurement then gives the FRC resolution, which is specified in the last column of Table~\ref{table:detailedResults}.
For both the FRC curves measures and the resolution measure, we again see a significant improvement in the transformer architectures compared to the single-image approaches.

\begin{figure*}
    \graphicspath{{figures/frcComparison}}
    \begin{subfigure}{0.33\textwidth}
        \input{figures/frcComparison/standard_comparison_gplt.tex}
        \caption{\label{fig:standardFrcComparison}}
    \end{subfigure}
    \begin{subfigure}{0.33\textwidth}
        \input{figures/frcComparison/directions_comparison_gplt.tex}
        \caption{\label{fig:directionFrcComparison}}
    \end{subfigure}
    \begin{subfigure}{0.33\textwidth}
        \input{figures/frcComparison/ablation_comparison_gplt.tex}
        \caption{\label{fig:ablationFrcComparison}}
    \end{subfigure}
    \caption{Comparison of Fourier ring correlation~(FRC) for different models. (a) Comparison between the standard models considered in Section~\ref{sec:results}. (b) Comparison between different numbers of viewing directions (see Table~\ref{table:moreClasses}). (c) Ablation study for the polar CNN and angular attention components (see Table~\ref{table:ablation}).}
\end{figure*}

\begin{figure*}
    \graphicspath{{figures/snrDropComparison}}
    \begin{center}
        \input{figures/snrDropComparison/attention_ss8_grpNrm4_errors_gplt}
    \end{center}
    \caption{\label{fig:mseComparison} The denoised relative MSE for various SNRs with Gaussian noise (but without CTF and shifts).
    The Wiener filter, DnCNN, U-Net, and polar CNN models denoise individual images, while the transformers denoise sets of images (the directional set contains eight images while the general set contains sixteen images split evenly over two viewing directions).}
\end{figure*}

\paragraph{SNR dependency}
To study the dependence of relative MSE for different SNRs, we train each model for noisy data at the pre-specified SNR and evaluate it on the training set.
For simplicity, these are performed without CTF or shifts.
The results are shown in Figure~\ref{fig:mseComparison}.
Here we have also included results for the covariance Wiener filter~\cite{Bhamre2016Denois} as implemented in ASPIRE~\cite{Aspire2025}.

We again see the improved performance of the polar transformer, but note that the relative improvement is greatest at high SNR.
The likely cause is that alignment becomes easier in this regime, allowing for better integration of information between images.
The general-set transformer also suffers more at low SNR, where the clustering step becomes more difficult and we see a larger gap in performance between it and the directional-set transformer.

Finally, we note that the performance of the single-image denoisers converge to that of the Wiener filter as the SNR drops.
The advantage of the non-linear properties of the neural network architectures is lost here and they essentially approximate the linear denoiser.

\begin{table*}
    \begin{center}
        \begin{tabular}{c|cccc}
            No. directions & Relative MSE (\textdownarrow) & PSNR (\textuparrow) & SSIM (\textuparrow) & FRC resolution (\textdownarrow) \\
            \hline
            1 & 0.042 & 32.2 & 0.97 & 2.68 \\
            2 & 0.049 & 31.6 & 0.97 & 2.81 \\
            4 & 0.051 & 31.5 & 0.96 & 2.92 \\
            8 & 0.055 & 31.2 & 0.96 & 2.93
        \end{tabular}
    \end{center}
    \caption{\label{table:moreClasses} Denoising results for the polar transformer in the general set case with varying number of viewing directions. In each case, there were eight images in each class.}
\end{table*}

\begin{table*}
    \begin{center}
        \begin{tabular}{c|cccc}
            Variant & Relative MSE (\textdownarrow) & PSNR (\textuparrow) & SSIM (\textuparrow) & FRC resolution (\textdownarrow) \\
            \hline
            Polar CNN + angular attention & 0.042 & 32.2 & 0.97 & 2.68 \\
            Polar CNN + standard attention & 0.068 & 30.2 & 0.96 & 3.32 \\
            Cartesian CNN + angular attention & 0.062 & 30.6 & 0.96 & 3.11 \\
            Cartesian CNN + standard attention & 0.093 & 28.8 & 0.94 & 3.51
        \end{tabular}
    \end{center}
    \caption{\label{table:ablation} Denoising results for ablated versions of the polar transformer architecture, swapping out both the processing layers and attention for ones that do not have the rotational equivariance.}
\end{table*}

\paragraph{Scaling the transformer}
In Section~\ref{sec:results}, we considered two regimes for the polar transformer: the directional set (where all images came from the same viewing direction, but with different in-plane rotations) and the general set (where the images came from two different viewing directions).
While for the directional set, the transformer simply had to align the images, for the general set, it also had to cluster the images (into one of two viewing directions).

A natural question is therefore how this generalizes to more viewing directions.
We have conducted additional experiments to investigate this, varying the number of viewing directions between one and eight (other conditions were kept the same, such as $\text{SNR} = 0.02$ and application of CTFs and shifts).
In each case, there were eight images for each viewing direction, so if the clustering is perfect, the denoised images should be of the same quality.
The results are shown in Table~\ref{table:moreClasses} (metrics) and Figure~\ref{fig:directionFrcComparison} (FRC curves).
We see that although there is some degradation in quality with more viewing directions, it is quite slight, suggesting that the polar transformer is able to distinguish between many viewing directions quite well.

\paragraph{Ablation study}
Another part of evaluating the proposed architecture is to replace parts with standard components and observe the difference in performance.
Here, we focus on three variants of the polar transformer architecture, applied to the directional-set problem with $\text{SNR} = 0.02$, including CTFs and shifts as in Section~\ref{sec:results}.
The first keeps the polar CNN components, but replaces angular attention with tthe standard attention mechanism, treating each image as a token to be combined with the others.
The second keeps the angular attention mechanism, but replaces the polar CNNs with Cartesian CNNs with the same configuration (in terms of channels).
Finally, the third replaces both the polar CNNs with Cartesian CNNs and the angular attention mechanism with standard attention.

The results are given in Table~\ref{table:ablation} (metrics) and Figure~\ref{fig:ablationFrcComparison}.
Here we see that each of these components (the polar CNN and the angular attention mechanism) is crucial to the superior performance of the polar transformer.
Note, however, how the angular attention mechanism seems to play a more important role since the Cartesian CNN + angular attention variant performs the best out of the three ablation variants.
This is natural as the Cartesian CNN is able to partially learn the rotational equivariance that is enforced by the architecture of the polar CNN.
In the above experiment, however, it is not enough to achieve a comparable performance.

\paragraph{Training on Parakeet data}
Although, as shown in \autoref{fig:denoisedExamples}, the models that were trained on simple synthetic projections already generalize to some extent to other data (like ones generated by Parakeet~\cite{Parkhurst2021Parakeet}), this has its limits.
Real-world micrographs usually have additional disturbing features (e.g. ice inhomogeneities) which, while less disruptive than the shot noise, can throw off a model that has not encountered such effects in training.

To demonstrate that this is a limitation of the training process rather than our model architecture, we have also trained models on projections generated with Parakeet at more aggressive / realistic settings.
\begin{figure*}
\begin{center}
 \includegraphics[width=\textwidth]{figures/parakeet-trained-examples.png}
\end{center}
    \caption{Qualitative results of denoising more challenging projections generated by Parakeet with a model that was also trained on Parakeet data. Top row: Parakeet simulations of images at electron dose \(100/\Angstr\) and defocus \(c_{10} = -5000\). Bottom row: simulation of the same pose, but with electron dose \(3\times10^9/\Angstr\) and no defocus. Middle row: estimates of the latter, using a polar transformer model with sets consisting of the top row projection and 7 similar projections each as input.}
 \label{fig:parakeetTrainedExamples}
\end{figure*}
On these projections, the previously discussed models would not perform well, but dedicated training makes it possible.
In \autoref{fig:parakeetTrainedExamples}, some examples are shown.
Training on such data brings our models close to being usable on real-world data.

\addtocounter{section}{1}
\section{More 3D reconstruction results}
\label{sec:reconstrReliability}
In Section~\ref{sec:reconstruction} it was demonstrated that the improved denoising our method provides can also lead to better 3D reconstructions.
Figure~\ref{fig:reconstrFSCproblematics} shows the FSCs for reconstructions of two other molecules: CENP-E (PDB ID 1T5C) and a MoaD protein (PDB ID 1V8C).

\begin{figure*}
\begin{center}
\begin{subfigure}[c]{0.65\textwidth}
 \includegraphics[width=0.32\textwidth]{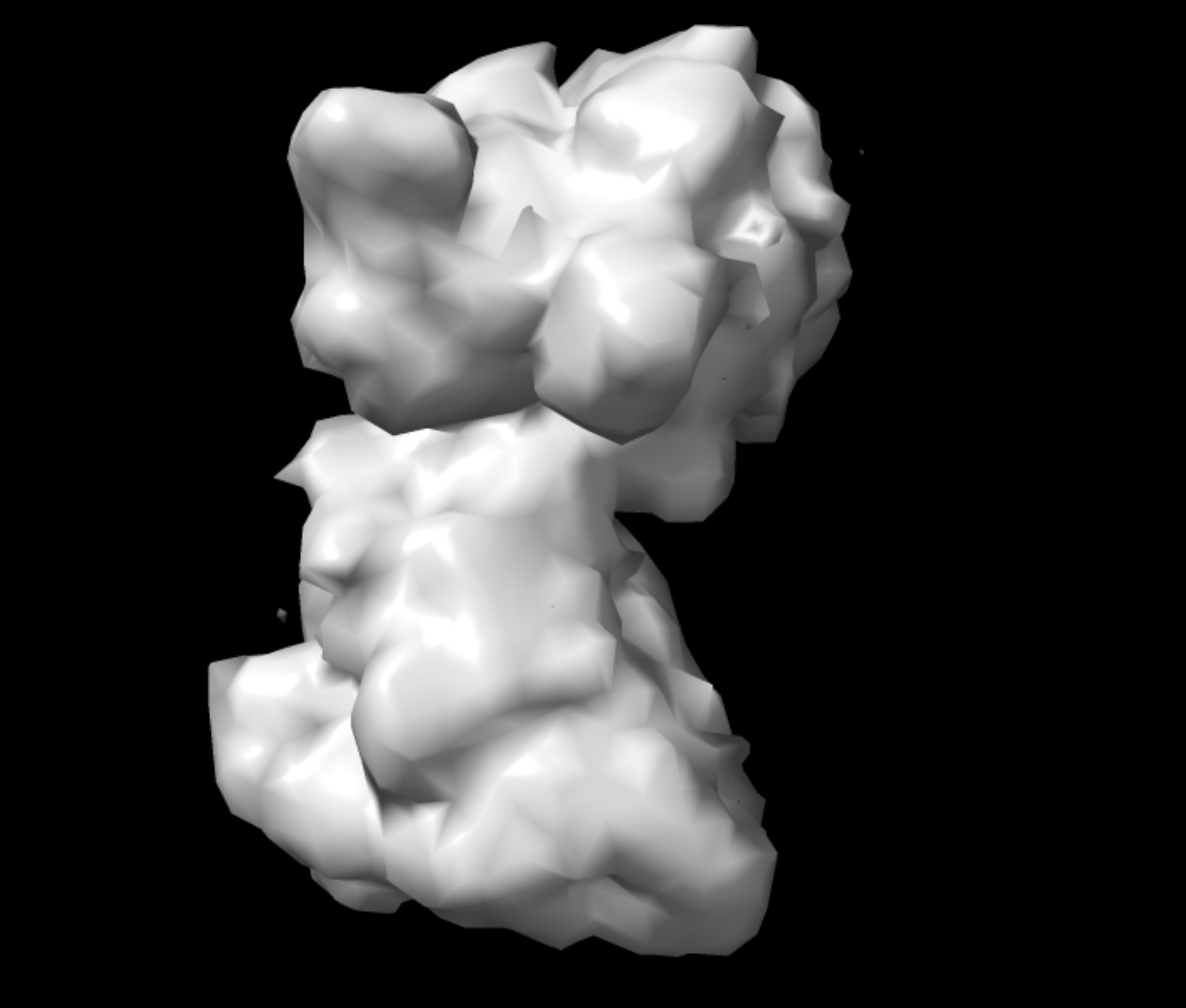}
 \includegraphics[width=0.32\textwidth]{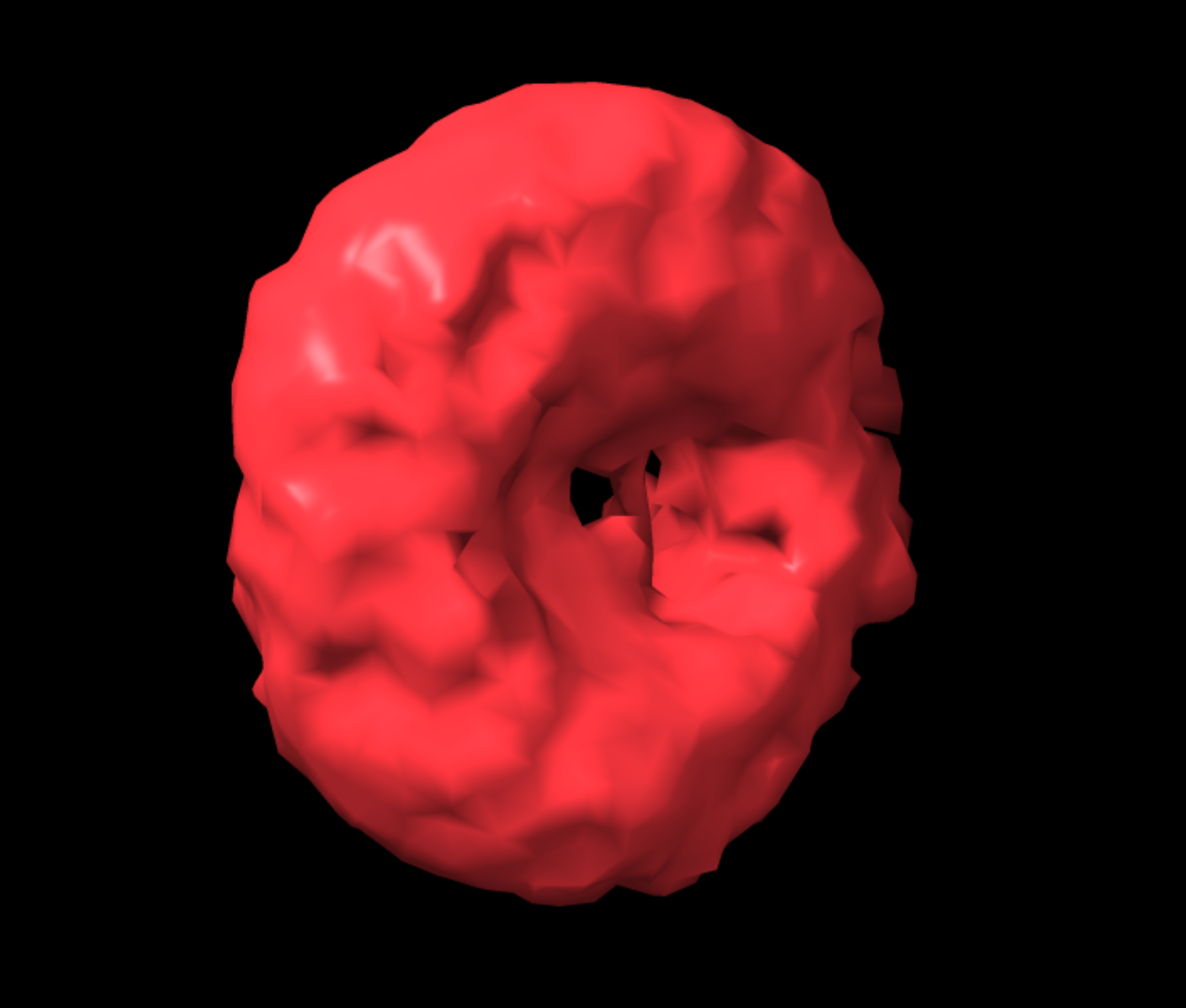}
 \includegraphics[width=0.32\textwidth]{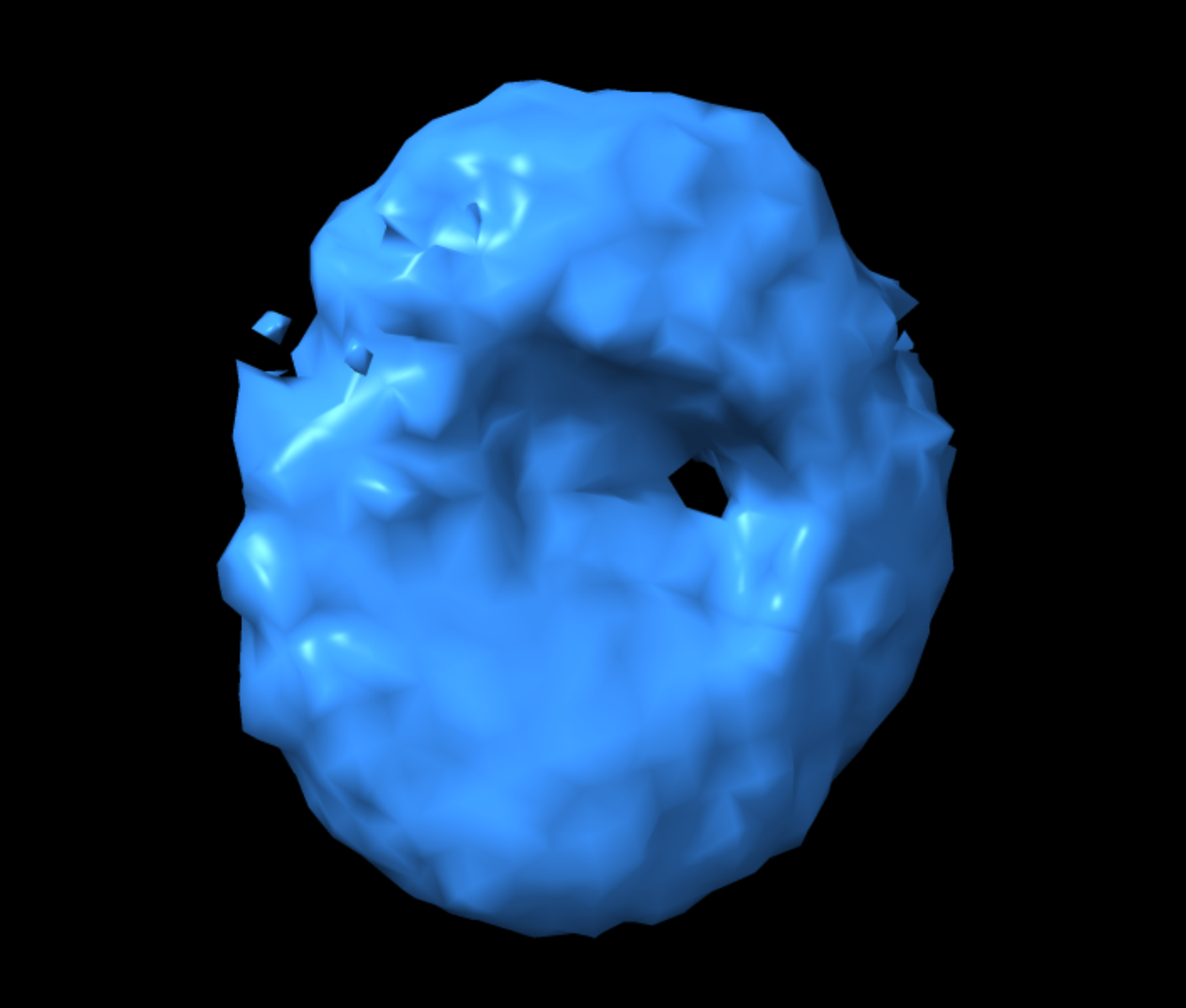}
\end{subfigure}
\begin{subfigure}[c]{0.33\textwidth}
 \includegraphics[width=\textwidth]{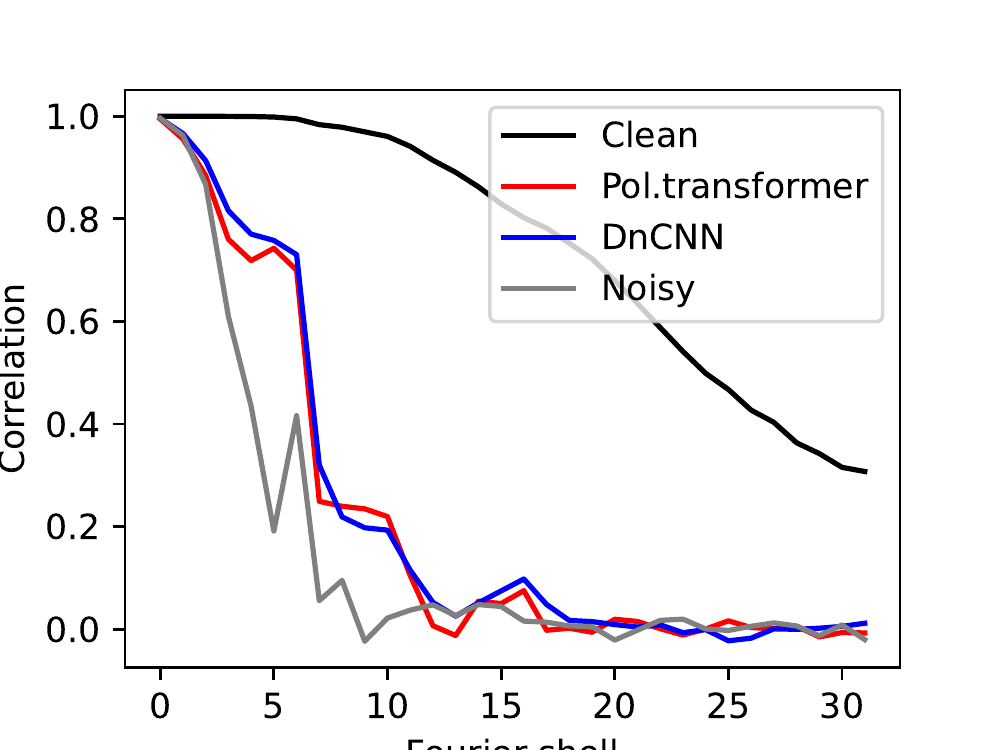}
\end{subfigure}
\begin{subfigure}[c]{0.65\textwidth}
 \includegraphics[width=0.32\textwidth]{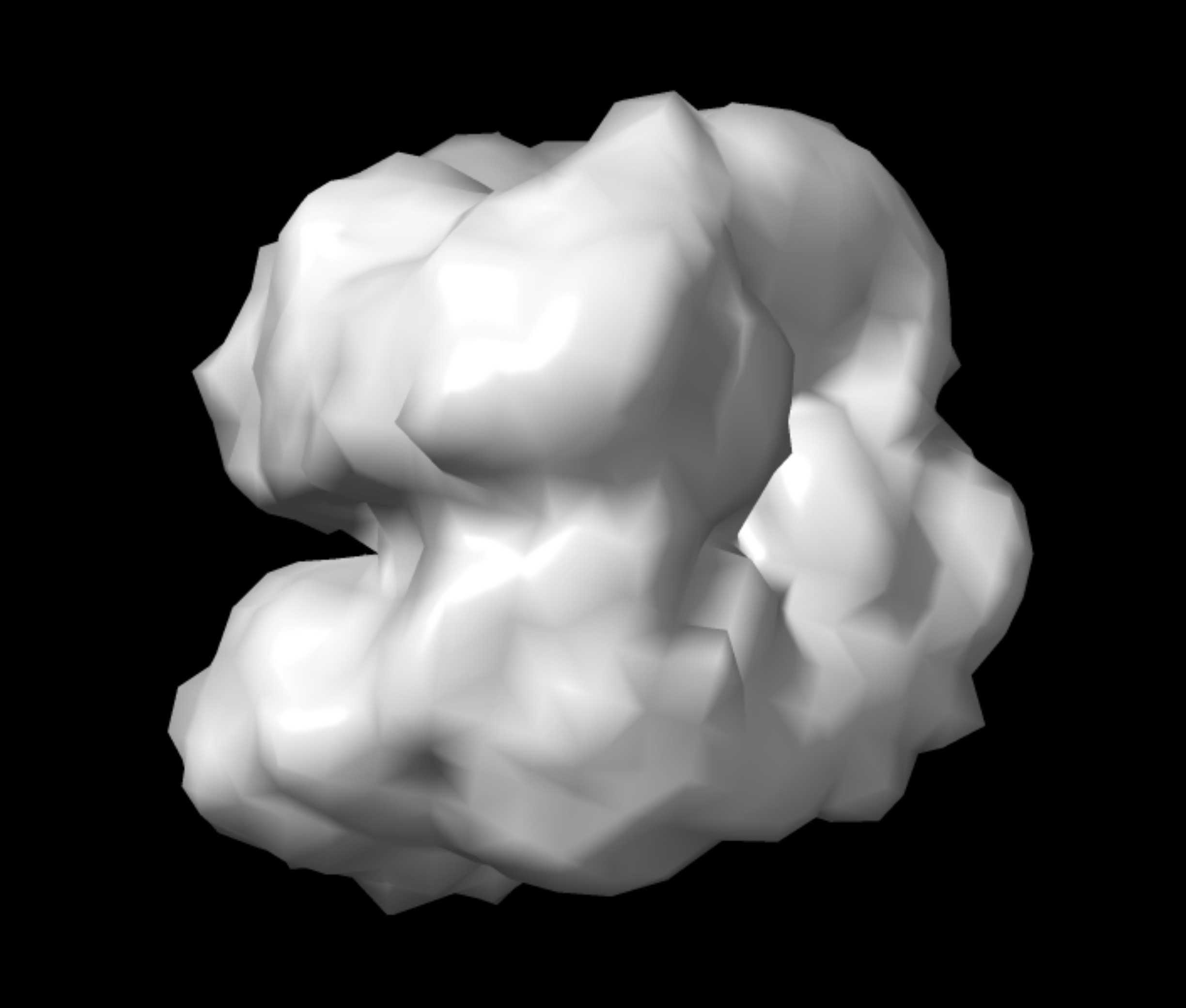}
 \includegraphics[width=0.32\textwidth]{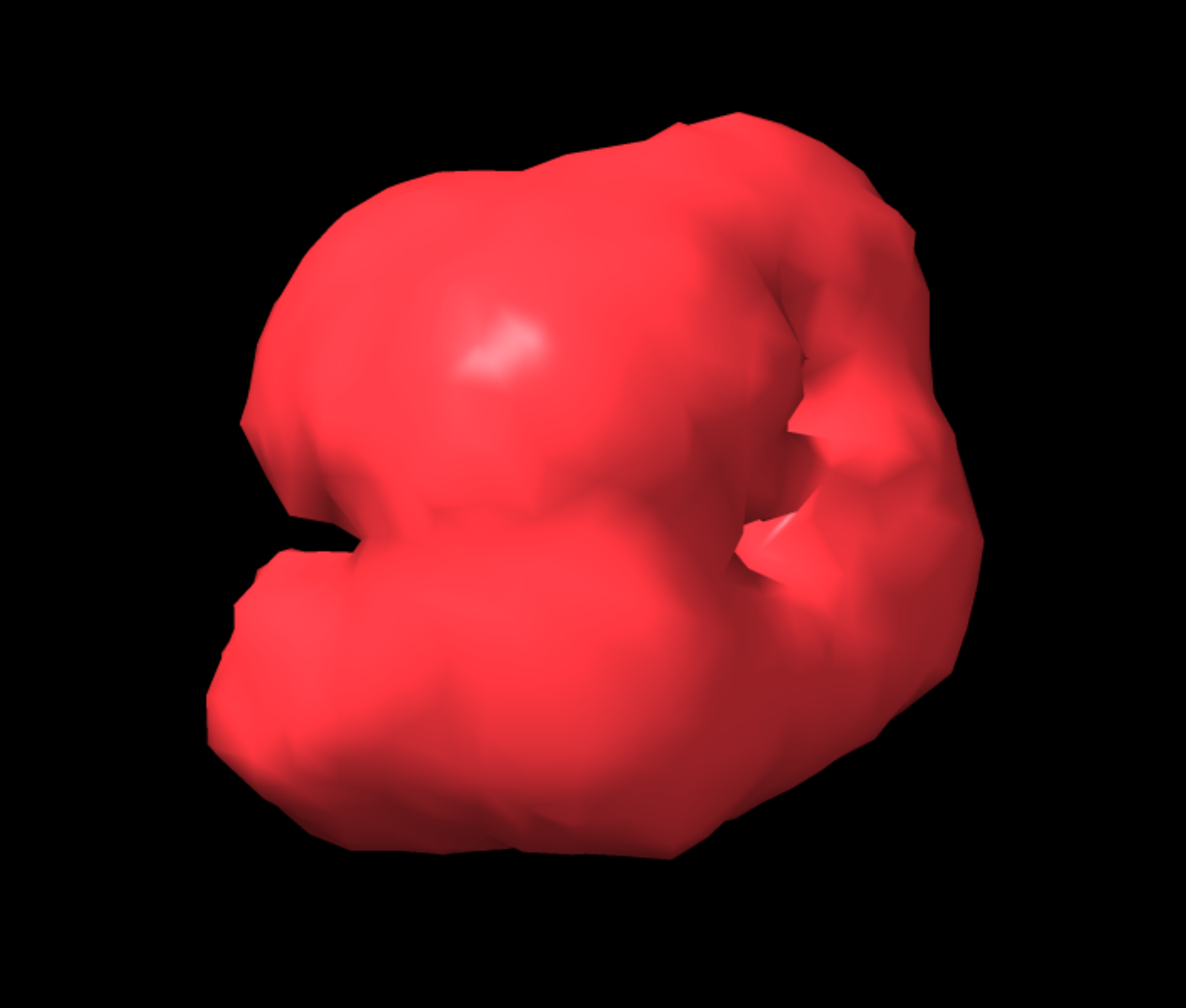}
 \includegraphics[width=0.32\textwidth]{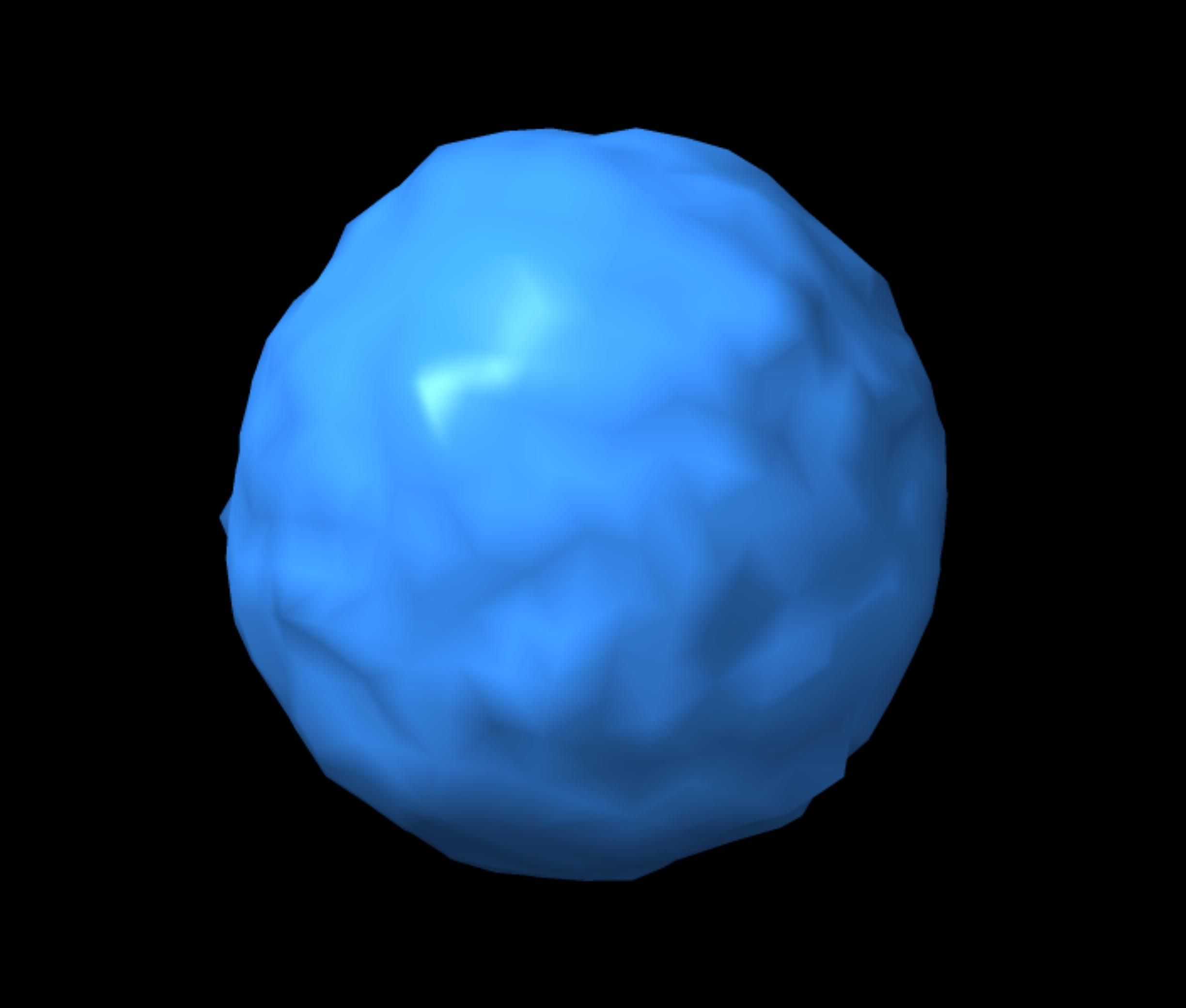}
\end{subfigure}
\begin{subfigure}[c]{0.33\textwidth}
 \includegraphics[width=\textwidth]{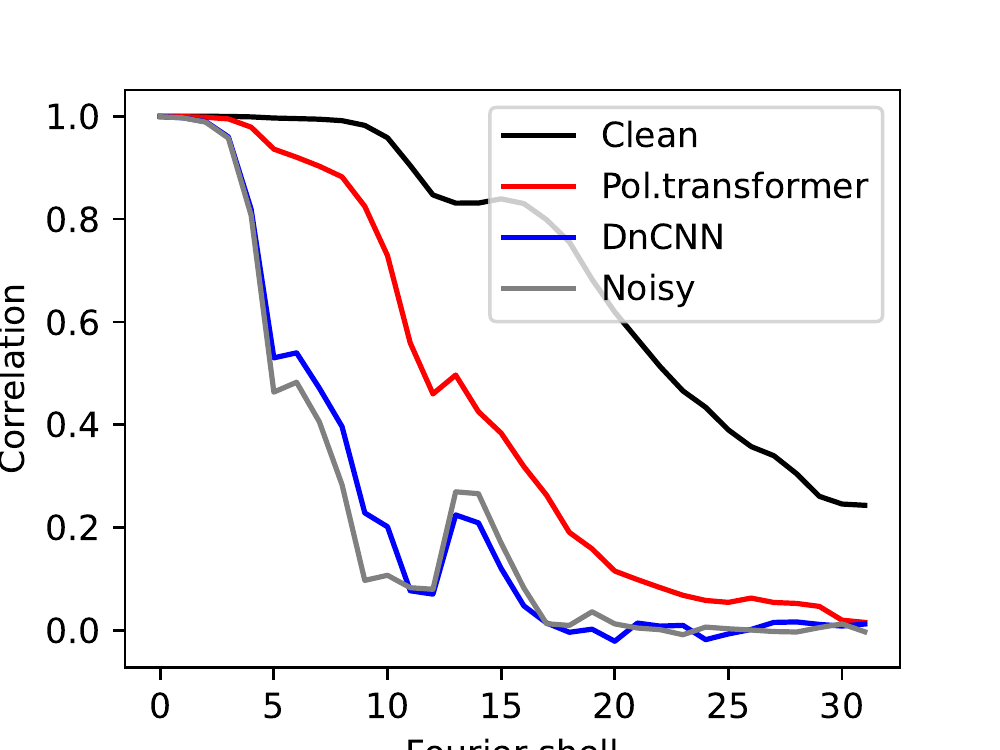}
\end{subfigure}
\end{center}
    \caption{Reconstruction quality like in Figure~\ref{fig:reconstrFSC}, but for molecules with PDB IDs 1T5C (top) and 1V8C (bottom).}
 \label{fig:reconstrFSCproblematics}
\end{figure*}

We note that for 1V8C, the polar transformer yields a more accurate reconstruction, but that 1T5C, the performance is similar to that of DnCNN.
There are a few potential reasons for this.
First, we are working with only 2 048 projections, which means that certain viewing directions may not be very well represented.
Second, the clustering algorithm may result in sets where images close to the center of the cluster have a viewing direction close to that of the others, but images close to the edge of the cluster have higher error.
As a result, those images will contribute to a lower-quality reconstruction.
Third, the clustering algorithm only operates on a hemisphere (combining antipodal viewing directions together), which creates a bias for viewing directons close to the equator.

We also note that these reconstructions are not truly ab initio, since prior knowledge of the viewing directions were used to cluster the images.
That being said, it is an approximation of what would be achieved using standard 2D classification methods, such as bispectrum-based approaches~\citep{Zhao2014Rotationally}.
It therefore provides an indication of the performance of 3D reconstruction from the denoised images.
A full reconstruction pipeline incorporating our method still requires improving several aspects, which will need to be addressed in future work.

\end{document}